\documentclass{article}
\usepackage{amssymb,amsmath,graphicx,ucs,hyperref}
\usepackage[utf8x]{inputenc}

\newtheorem{proposition}{Proposition}
\newtheorem{theorem}{Theorem}
\newtheorem{definition}{Definition}
\newtheorem{lemma}{Lemma}

\newenvironment{proof}{\noindent \emph{Proof. }}{\hfill \hbox{\rlap{$\sqcap$}$\sqcup$}\\}

\title{Density of Binary Disc Packings:\\The $9$ Compact Packings}

\author{
Nicolas Bédaride
\footnote{
I2M Université Aix Marseille, centrale, CNRS UMR7373, 13453 Marseille, France.}
\and
Thomas Fernique
\footnote{LIPN, Université Paris Nord, CNRS UMR7030, 93430 Villetaneuse, France.}
}
\date{}

\begin{document}
\maketitle

\begin{abstract}
A disc packing in the plane is {\em compact} if its contact graph is a triangulation.
There are $9$ values of $r$ such that a compact packing by discs of radii $1$ and $r$ exists.
For each of these $9$ values, we prove that the maximal density over all the packings by discs of radii $1$ and $r$ is reached for a compact packing.
We describe such a packing and give its density.
\end{abstract}

\section{Introduction}

A {\em disc packing} (or {\em circle packing}) is a set of interior-disjoint discs in the Euclidean plane.
Its {\em density} $\delta$ is the proportion of the plane covered by the discs:
\begin{displaymath}
\delta:=\limsup_{k\to \infty}\frac{\textrm{area of the square $[-k,k]^2$ covered by the discs}}{\textrm{area of the square $[-k,k]^2$}}.
\end{displaymath}
A central issue in packing theory is to find the maximal density of disc packings.

If the discs have all the same radius, it was proven in \cite{FT43} that the density is maximal for the {\em hexagonal compact packings}, where discs are centered on a suitably scaled triangular grid (see also \cite{CW10} for a short proof).

For {\em binary} disc packings, {\em i.e.}, packings by discs of radii $1$ and $r\in(0,1)$ where both disc sizes appear, there are only seven values of $r$ for which the maximal density is known \cite{Hep00,Hep03,Ken04}.
These values are specific algebraic numbers which allow a {\em compact packing}, that is, a packing whose {\em contact graph} (the graph which connects the centers of any two tangent discs) is a triangulation.
In each of these seven cases, the maximal density turns out to be reached for a compact disc packing.
Compact packings thus seem to be good candidates to maximize the density.

It was proven in \cite{Ken06} that there are exactly $9$ values $r$ which allow a binary compact packing by discs of radii $1$ and $r$: the seven above mentioned, and two other ones ($r_5$ and $r_9$ in Fig.~\ref{fig:targets}).
We here prove that compact packings also maximize the density for these two remaining values.
We actually provide a new self-contained proof for all the $9$ values, denoted by $r_1,\ldots,r_9$:

\begin{theorem}
\label{th:main}
For each $r_i$ allowing a binary compact packing by discs of radius $1$ and $r_i$, the 
density of any packing by discs of radius $1$ and $r_i$ is less than or equal to the density $\delta_i$ of the periodic compact packing $\mathcal{P}_i$ depicted in Fig.~\ref{fig:targets}.
\end{theorem}

\begin{figure}[hbt]
\centering
\begin{tabular}{lll}
  $r_1\approx 0.63$\hfill $\delta_1\approx 0.9106$ & $r_2\approx 0.54$\hfill $\delta_2\approx 0.9116$& $r_3\approx 0.53$\hfill $\delta_3\approx 0.9141$\\
  \includegraphics[width=0.3\textwidth]{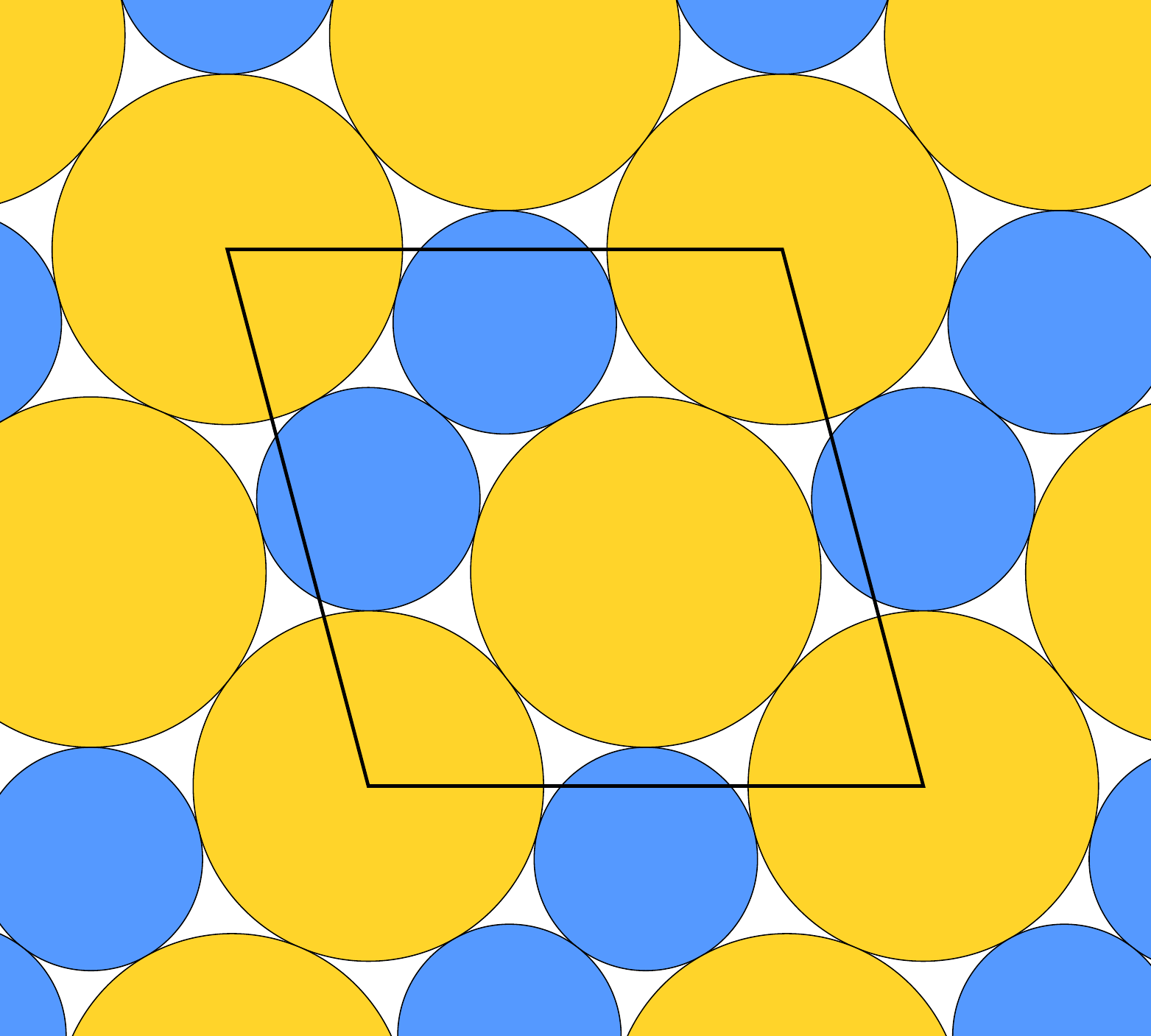} &
  \includegraphics[width=0.3\textwidth]{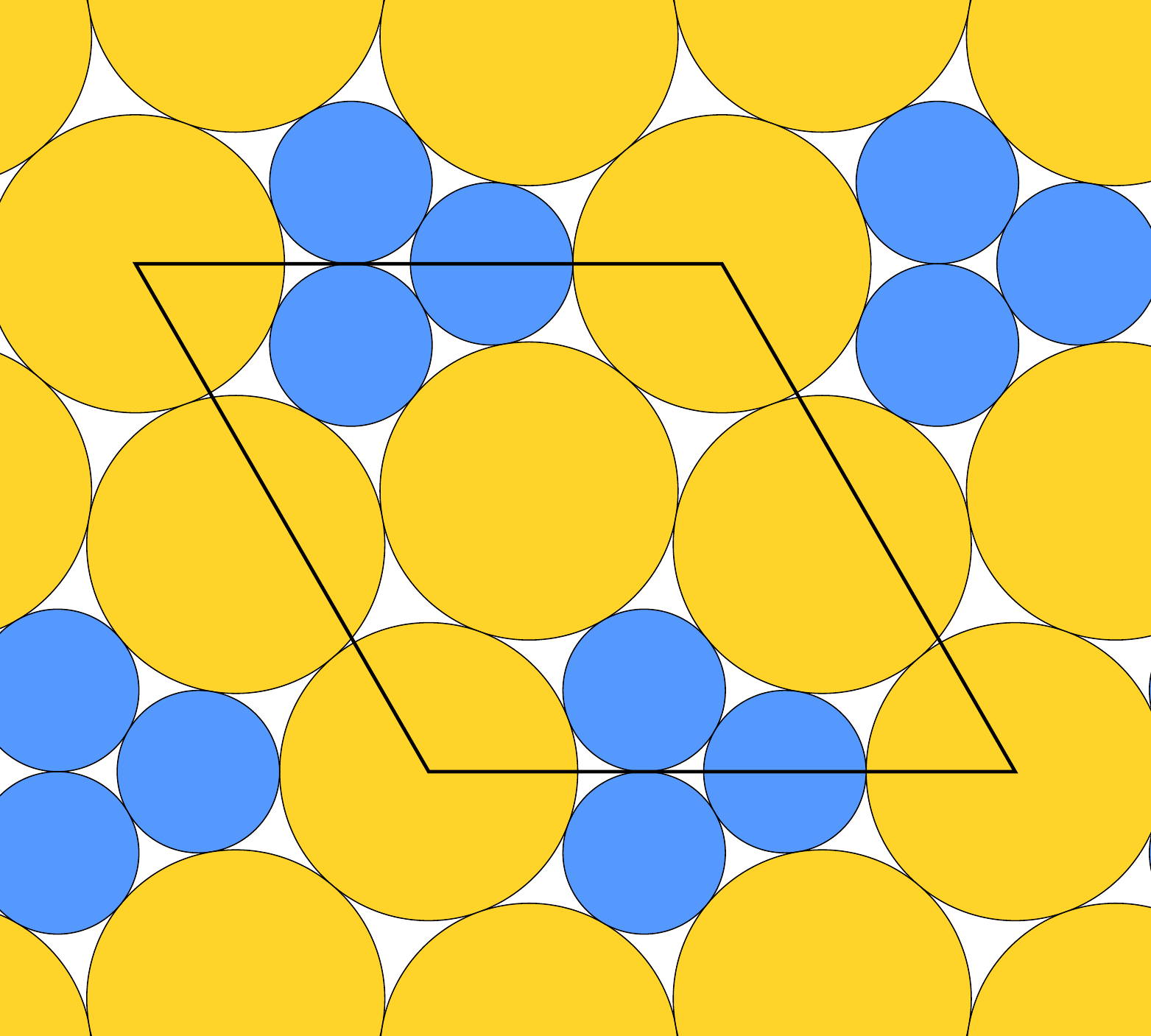} &
  \includegraphics[width=0.3\textwidth]{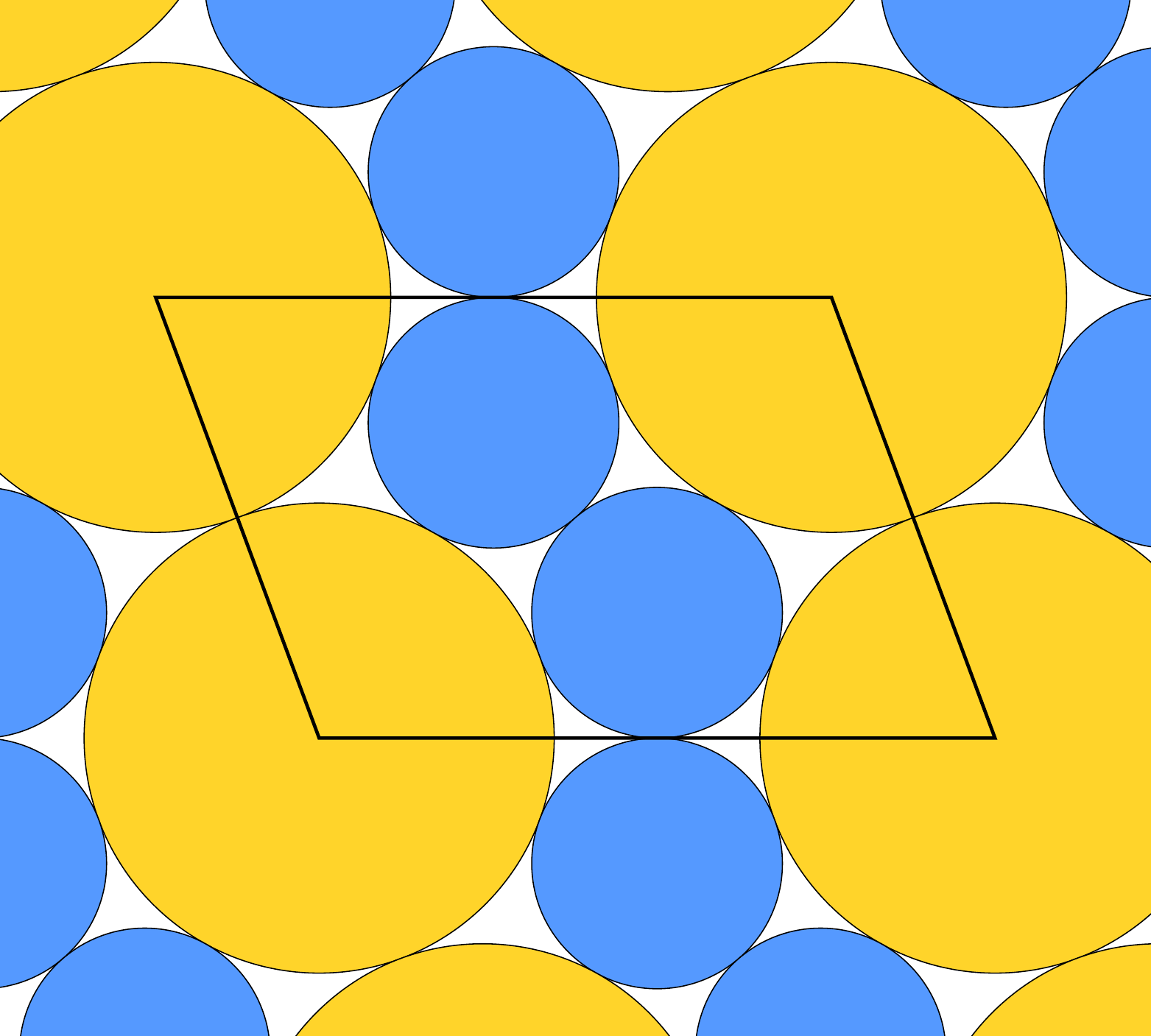}\\
  $r_4\approx 0.41$\hfill $\delta_4\approx 0.9201$ & $r_5\approx 0.38$\hfill $\delta_5\approx 0.9200$ & $r_6\approx 0.34$\hfill $\delta_6\approx 0.9246$\\
  \includegraphics[width=0.3\textwidth]{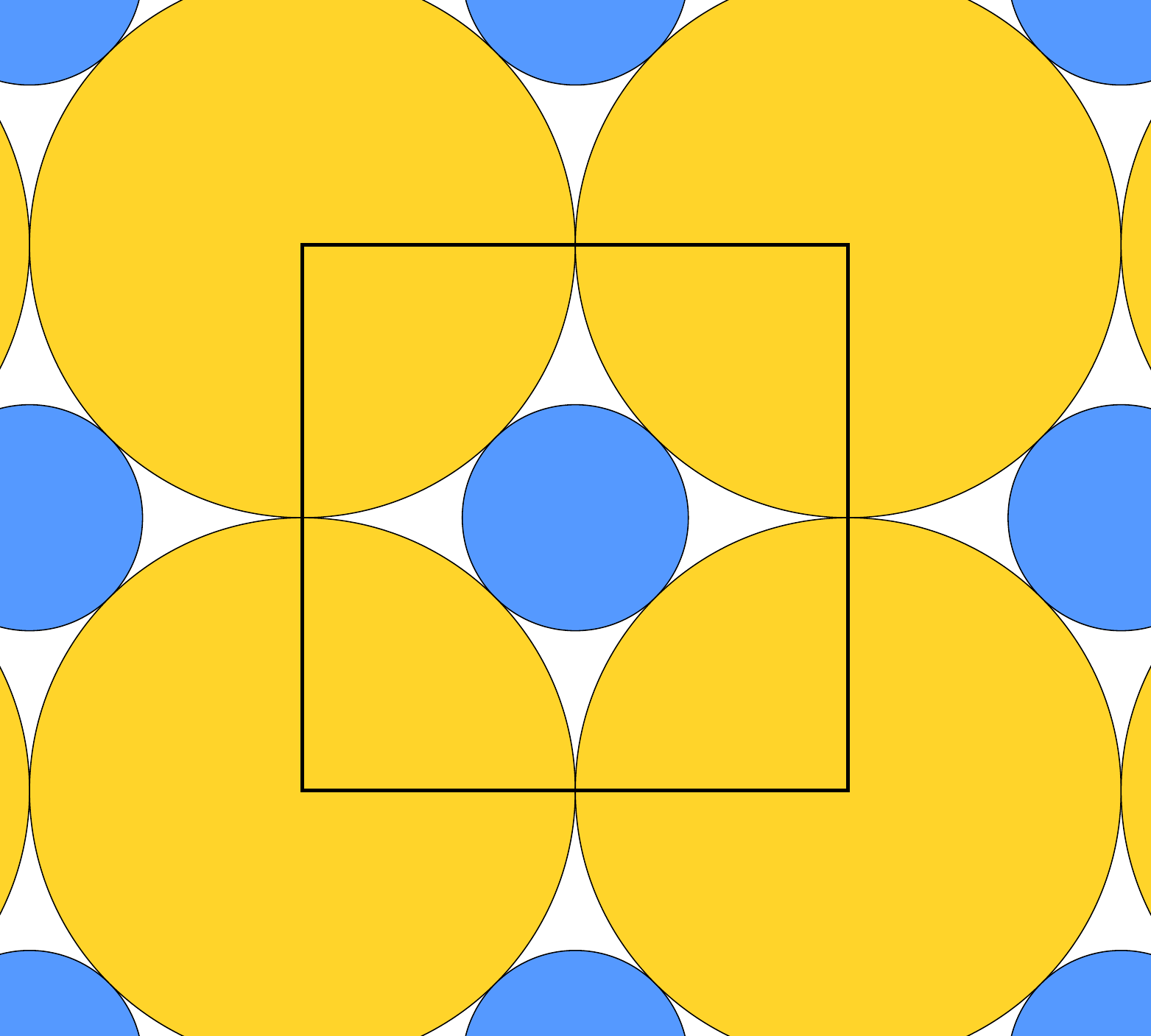} &
  \includegraphics[width=0.3\textwidth]{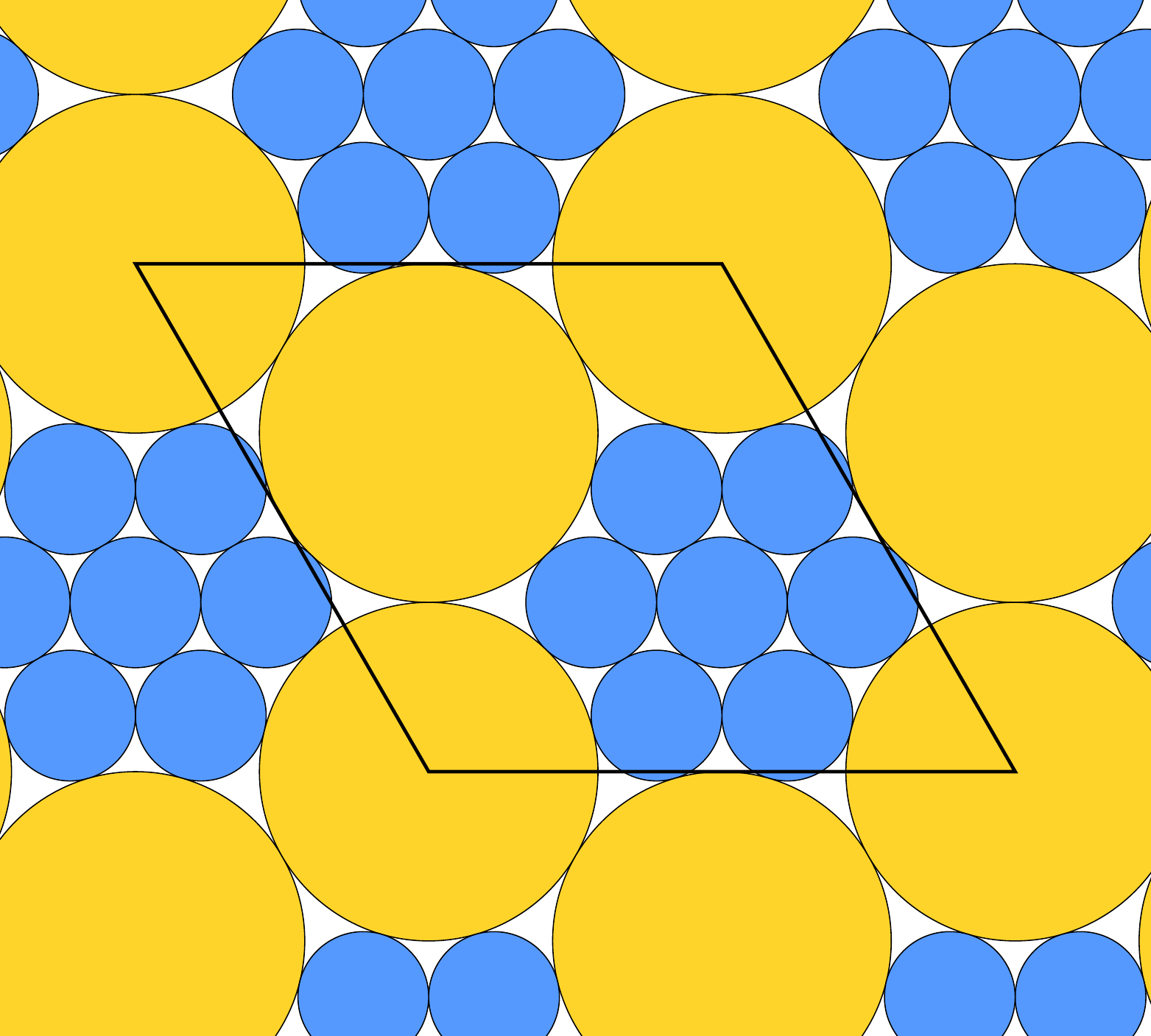} &
  \includegraphics[width=0.3\textwidth]{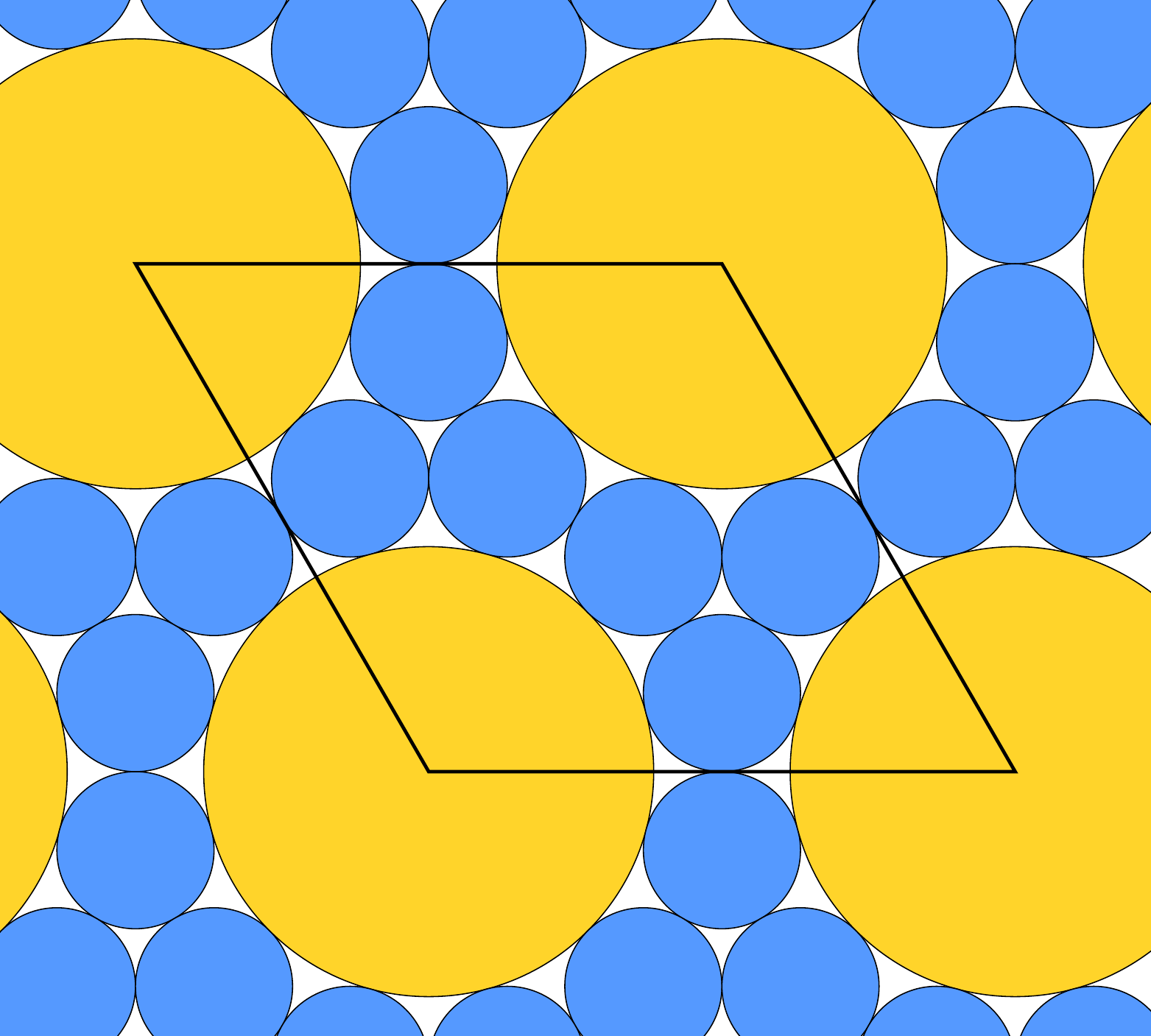}\\
  $r_7\approx 0.28$\hfill $\delta_7\approx 0.9319$ & $r_8\approx 0.15$\hfill $\delta_8\approx 0.9503$ & $r_9\approx 0.10$\hfill $\delta_9\approx 0.9624$\\
  \includegraphics[width=0.3\textwidth]{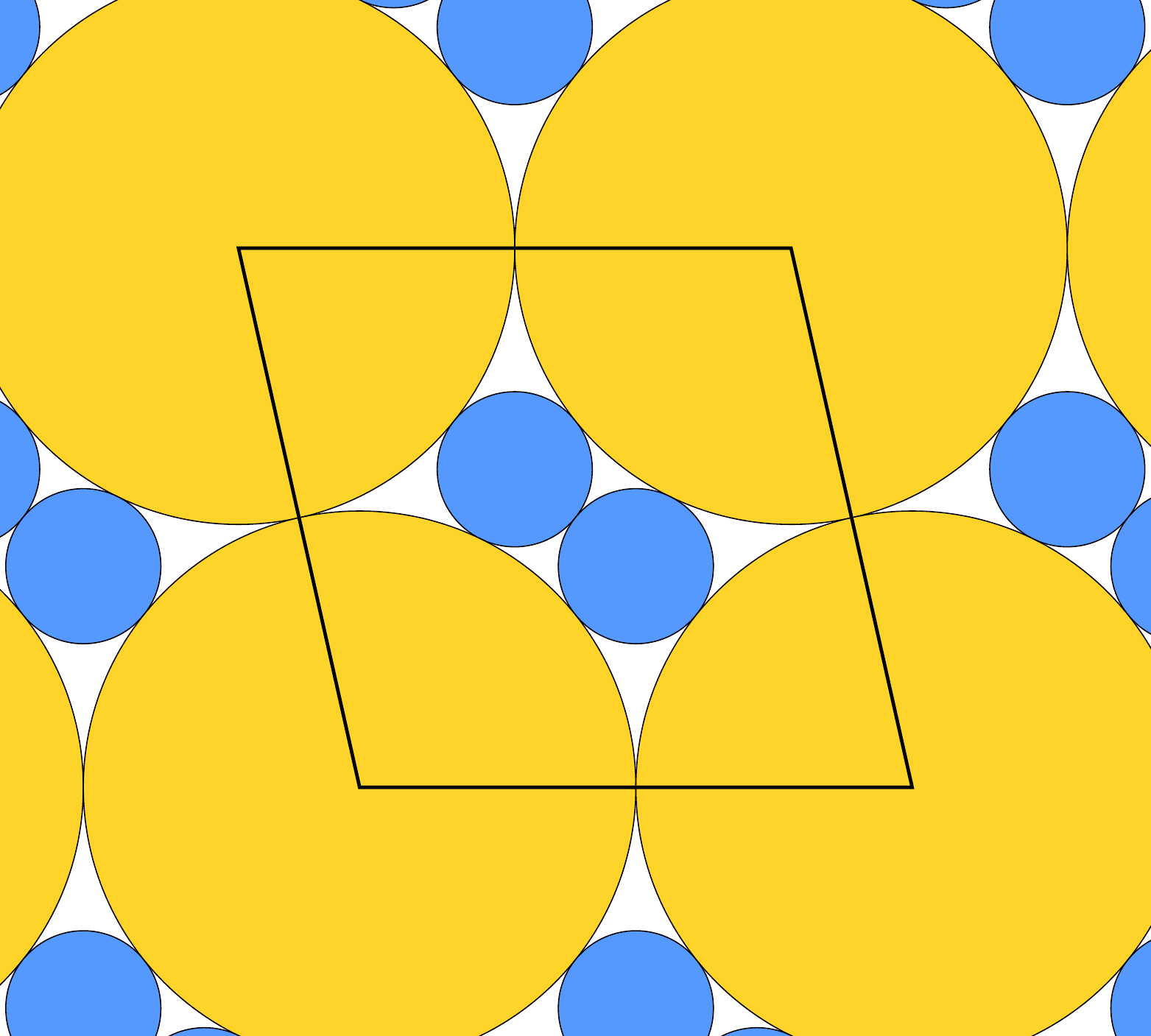} &
  \includegraphics[width=0.3\textwidth]{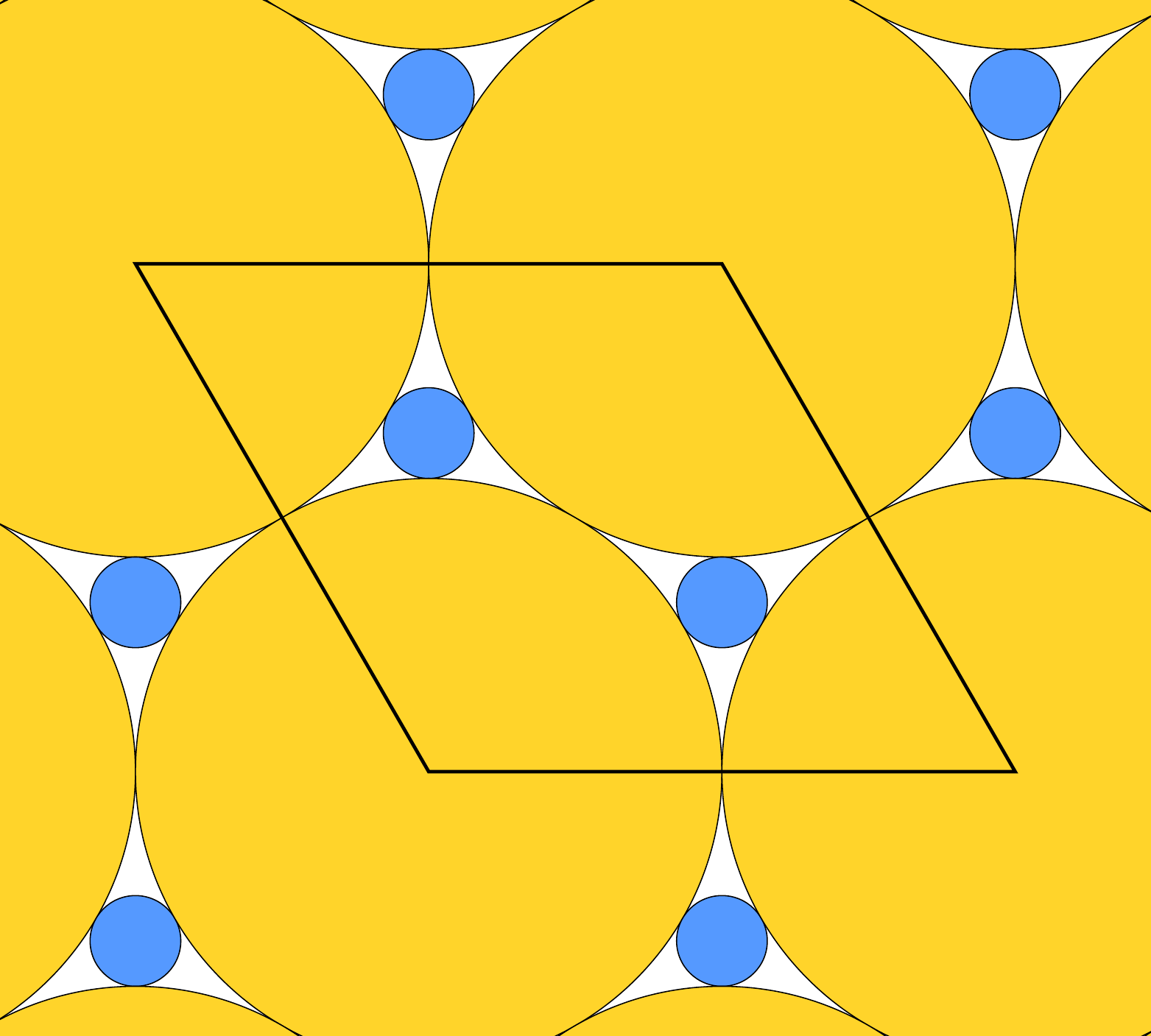} &
  \includegraphics[width=0.3\textwidth]{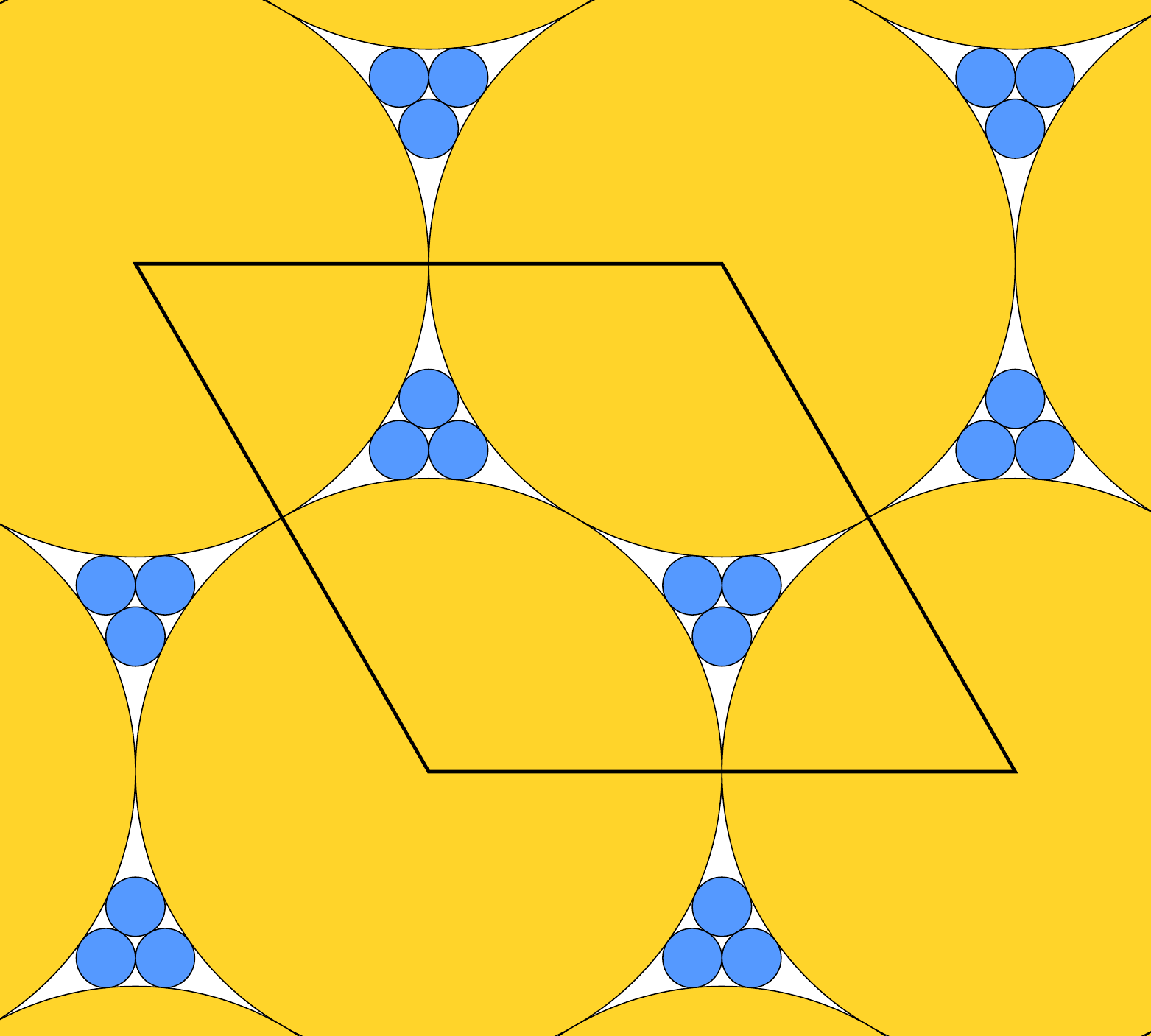}
\end{tabular}
\caption{
The packings $\mathcal{P}_1,\ldots,\mathcal{P}_9$ which maximize the density.
They are periodic, with the black parallelogram showing a fundamental domain (each picture is scaled so that its area is one quarter of the picture area).
Numerical approximations by truncation of the ratio $r_i$ of the disc radii and of the density $\delta_i$ of the packing are given (both $r_i$ and $\delta_i/\pi$ are algebraic: their minimal polynomials are given in Appendix~\ref{sec:exact}).
}
\label{fig:targets}
\end{figure}

Proving Theorem~\ref{th:main} amounts to solve an optimization problem.
Since discs have real coordinates, there is a continuum of cases to consider.
Such an issue already arises for packings restricted to bounded regions of the plane.
In this context, {\em interval arithmetic} is a powerful tool to optimize a function over a compact set with the help of a computer.
This is used, for example, in \cite{FKS19} to find the smallest disc that can contain a packing of any set of discs whose total area is fixed.
We shall also use it here, as detailed in Section~\ref{sec:computer}.

An even more serious problem is that the coordinates of the circles are not bounded since we consider packings of the whole Euclidean plane.
Can we get back to an optimization problem over a compact set?
Or, with the formalism introduced in \cite{Lag02}, can we define suitable ``weighting rules'' which brings the problem back to solving ``local density inequalities''?
Although there is no general guarantee that this is possible, this is the strategy successfully followed in \cite{Hal05} to prove the Kepler conjecture or in \cite{Hep00,Hep03,Ken04} to prove the maximal density of $7$ of the $9$ packings in Fig.~\ref{fig:targets}.
This is also the strategy that we here follows, as outlined in Section~\ref{sec:localization}.

In the light of Theorem~\ref{th:main} (and the fact that with a single disc size, the density is maximized by the hexagonal compact packing), we can legitimately wonder if compact packings always maximize density among packings with the same disc sizes.
This was conjectured in \cite{CGSY18} to holds under an additional {\em saturation hypothesis} (a packing is saturated if no more discs can be added).
In particular, what about the $164$ compact packings with three sizes of discs discovered in \cite{FHS20}?
The conjecture was proven to hold for one case in \cite{Fer19}.
However, it was later disproven on another case in \cite{FP21} (where a case which illustrates the necessity of the saturation hypothesis is also provided).
Most of the other cases remain open and maybe an additional hypothesis could correct the conjecture.

\section{Strategy}
\label{sec:localization}

Given a disc packing, we shall decompose the plane by a specific triangulation $\mathcal{T}$ of the disc centers, with triangles of uniformly bounded diameter (Section~\ref{sec:triangulation}).
We define the {\em emptiness} $E(T)$ of a triangle $T\in\mathcal{T}$ by
\begin{displaymath}
E(T):=\delta_i\cdot\textrm{area}(T)-\textrm{cov}(T),
\end{displaymath}
where $\textrm{area}(T)$ is the area of $T$, $\textrm{cov}(T)$ is the area of $T$ inside the discs centered on the vertices of $T$ and $\delta_i$ is the density of the $i$-th {\em target packing} $\mathcal{P}_i$ (Fig.~\ref{fig:targets}).
A triangle has thus positive emptiness if it has density less than $\delta_i$, and proving than the packing has overall density $\delta\leq\delta_i$ amounts to prove that the average emptiness is nonnegative, since
\begin{displaymath}
\limsup_{k\infty}\frac{1}{4k^2}\sum_{T\in \mathcal{T}_k}E(T)=\delta_i-\delta,
\end{displaymath}
where $\mathcal{T}_k$ denotes the triangles of $\mathcal{T}$ which intersect the square $[-k,k]^2$.
Indeed, since the triangles are uniformly bounded, the total area (hence the emptiness) of the triangles which cross the boundary of $[-k,k]^2$ is $O(k)$ and we have:
$$
\sum_{T\in \mathcal{T}_k}E(T)=\delta_i\cdot\left(\textrm{area of $[-k,k]^2$}\right)-\left(\textrm{area of $[-k,k]^2$ covered by discs}\right)+O(k).
$$
Dividing by the area of $[-k,k]^2$ and taking the limit yields the claimed equality.

In order to prove that the average emptiness is nonnegative, we will define over triangles a {\em potential} $U$ which satisfies two inequalities.
The first one, further referred to as {\em the global inequality}, involves all the triangles of $\mathcal{T}$:
\begin{equation}
\label{eq:global}
\limsup_{k\infty}\frac{1}{4k^2}\sum_{T\in \mathcal{T}_k}U(T)\geq 0.
\end{equation}
The second one, further referred to as {\em the local inequality}, involves any triangle $T$ which can appear in $\mathcal{T}$:
\begin{equation}
\label{eq:local}
E(T)\geq U(T).
\end{equation}
The result then trivially follows:
\begin{displaymath}
\limsup_{k\infty}\frac{1}{4k^2}\sum_{T\in \mathcal{T}_k}E(T)\geq\limsup_{k\infty}\frac{1}{4k^2}\sum_{T\in \mathcal{T}_k}U(T)\geq 0.
\end{displaymath}
Since the global inequality for $U$ is the same as for $E$, it seems we just made things worse by adding a second inequality.
However, we shall choose $U$ so that we can prove that the global inequality follows from an inequality on a finite set of finite configurations.
Namely, the potential of a triangle $T$ will be the sum of a {\em vertex potential} $U_v(T)$ defined on each vertex $v\in T$ and an {\em edge potential} $U_e(T)$ defined on each edge $e\in T$
$$
U(T):=\sum_{v\in T}U_v(T)+\sum_{e\in T}U_e(T),
$$
such that, for any vertex $v$ and edge $e$ of the triangulation $\mathcal{T}$:
\begin{equation}
\label{eq:vertex}
\sum_{T\in\mathcal{T}|v\in T} U_v(T)\geq 0
\end{equation}
\begin{equation}
\label{eq:edge}
\sum_{T\in\mathcal{T}|e\in T} U_e(T)\geq 0.
\end{equation}
Inequality~\eqref{eq:vertex}, which involves the triangles sharing a vertex, is proven in Section~\ref{sec:global_vertex} (Prop.~\ref{prop:vertex_pot}).
Inequality~\eqref{eq:edge}, which involves pairs of triangles sharing an edge, is proven in Section~\ref{sec:global_edge} (Prop.~\ref{prop:edge_pot}).

The local inequality~\eqref{eq:local} has then to be proven for each triangle of the decomposition.
We make two cases, depending whether the triangle is a so-called {\em $\varepsilon$-tight} triangle or not.
The former case, considered in Section~\ref{sec:local_tight}, is proven with elementary differential calculus (Prop.~\ref{prop:local_tight}).
The latter case, considered in Section~\ref{sec:local}, is proven with a computer by dichotomy (Prop.~\ref{prop:local}).
Theorem~\ref{th:main} follows.

\section{Computer use}
\label{sec:computer}

The first use we will make of the computer is to check inequalities on real numbers.
Since real numbers cannot be all exactly represented on a computer, we shall systemically use {\em interval arithmetic}.
The principle is simple:
\begin{itemize}
\item any real number is represented on the computer by an interval whose endpoints are exactly representable floating-point numbers.
\item computations are performed in a conservative way, that is, the image by an $n$-ary real function $f$ of intervals $x_1,\ldots,x_n$ must be an interval which contains at least all the real numbers $f(y_1,\ldots,y_n)$ for $y_i\in x_i$.
\end{itemize}
The interval given by a sequence of computation will thus always contain at least the exact result.
Exactly what this interval is depends on the actual software implementation.
In this paper, we use SageMath \cite{sage} with its $53$ bits of precision interval arithmetic.
Then, to prove an inequality $A\leq B$ between real number $A$ and $B$, we will check that the right endpoint of the interval which represents $A$ is less than or equal to the left endpoint of the interval which represents $B$.
If this checking fails, it does not means that the inequality is false but that we cannot prove it with interval arithmetic.
This is how inequalities~\eqref{eq:local}, \eqref{eq:vertex} and \eqref{eq:edge} are proven.

In this article, numerical values are given in Tables~\ref{tab:m}, \ref{tab:m_c5}, \ref{tab:Z}, \ref{tab:eps} and \ref{tab:LQ}.
Despite their appearance as approximate values, these values are exact in the mathematical sense and can be used by anyone to get the claimed results.
However, someone who is interested in computer arithmetic will be surprised that these values are not exactly representable in floating-point arithmetic, even though they have been found by a computer program.
How can they be used on the computer to obtain the claimed results?
The point is that giving exactly representable values in an article would have no real practical interest: who is going to copy these long sequences of number to check with his own program the validity of our results?
Instead, we decided to give simpler values, usually not exactly representable, but with the guarantee that any value whose first digits coincide with all the given ones can be used in a computer check.
For example, in the first line and second column of Table~\ref{tab:m}, instead of the lower bound
$$
0.0004908434532219589619371491462374024195014499127864837646484375,
$$
which is an exact representable value, we give the value $0.0005$.
Since it is larger, it is still a mathematically valid lower bound.
It is not exactly representable, but the above-mentioned guarantee ensures that any value between $0.0005$ and $0.0006$ is valid and can be used in a computer check (this have been ensured by computing with an interval containing $[0.0005, 0.0006]$).
In conclusion, the given values have the advantage of being ``humanly readable'', mathematically valid, and usable by a reader who would like to check the calculations with his own program.
This also shows that there is in fact some margin on the exact value of the constants used in the code.

The second use we will make of the computer is to check inequalities over intervals (namely products of intervals, in the proof of Prop.~\ref{prop:local}).
We again use interval arithmetic for this, but instead of using an interval to represent an exact real number, we simply use the interval itself!
In other words, to prove an inequality $f(x)\leq g(x)$ over an interval $[a,b]$, we compute $f(x)$ and $g(x)$ in interval arithmetic with $x=[a,b]$ and we check whether the right endpoints of the interval $f(x)$ is less than or equal to the left endpoints of the interval $g(x)$.

The last use we will make of the computer is to perform exhaustive search.
This is done in a classic way in the proof of Prop.~\ref{prop:vertex_pot}, where there is a finite number of cases to check, as well as in a less conventional way in the proof of Prop.~\ref{prop:local}, where we check the local inequality~\eqref{eq:local} on an infinite set of triangles.
For that, we shall prove that the set of triangles is compact and break it down in finitely many sets which are sufficiently small so that interval arithmetic can be used (as explained above) to prove that {\em any} triangle in such a set satisfies the wanted inequality.

The complete code needed to verify all the results of this paper, around $400$ lines long, can be found in the supplementary materials (\verb+binary.sage+).
The language is python, with some use of the open-source software SageMath \cite{sage} ({\em e.g.} for interval arithmetic).
Although quite slow, Python has the advantage of being easily readable and (currently) quite widespread.
All the computations have been performed on our modest laptop, an Intel Core i5-7300U with $4$ cores at $2.60$GHz and $16$ Go RAM.

\section{FM-triangulation}
\label{sec:triangulation}

Given a disc packing, define the {\em cell} of a disc as the set points of the plane which are closer to this disc than to any other (Fig.~\ref{fig:FM_triangulation}, left).
These cells form a partition of the plane whose dual is a triangulation, referred to as the {\em FM-triangulation} of the packing (Fig.~\ref{fig:FM_triangulation}, right).
Actually, the dual may be not a triangulation if more than three cells meet in a point (this happens if identical circles are centered on the vertices of a regular $n$-gon, $n>3$).
In this case, we arbitrarily add any maximal choice of non-crossing chords.

\begin{figure}[hbt]
\centering
\includegraphics[width=0.49\textwidth]{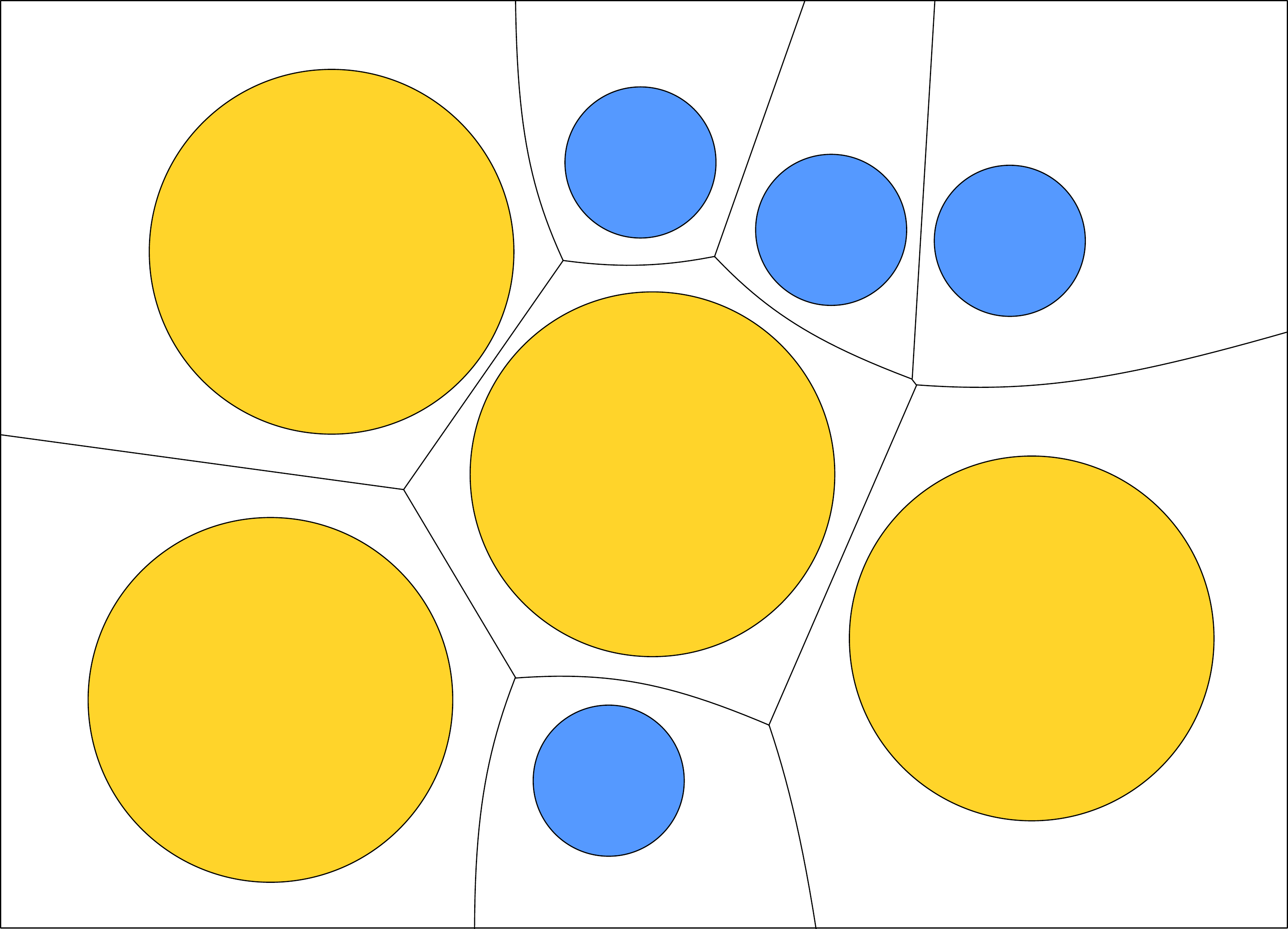}
\hfill
\includegraphics[width=0.49\textwidth]{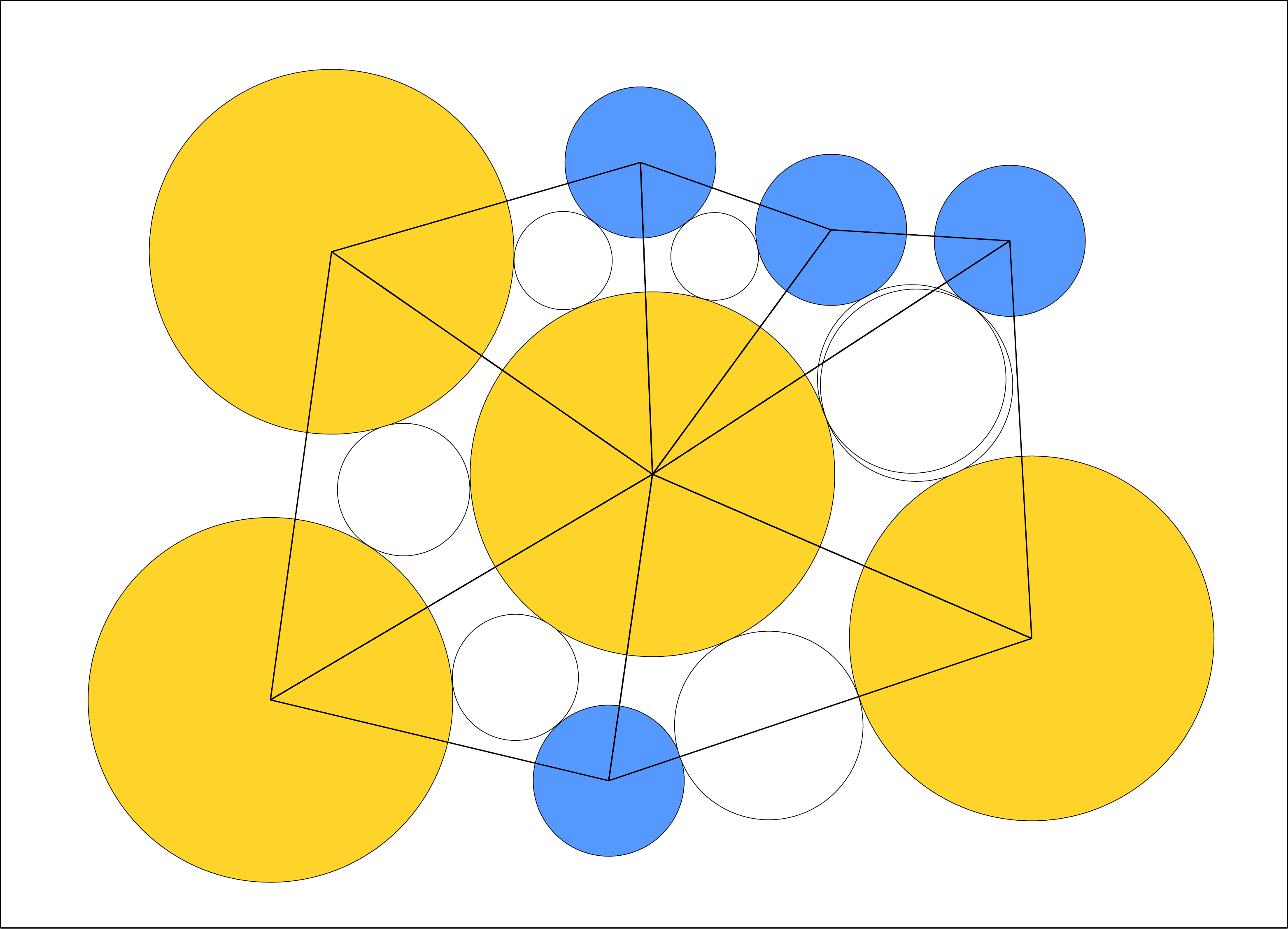}
\caption{
Some discs and their cells (left).
The corresponding FM-triangulation with the support discs in white (right).
Each support disc is centered on a vertex of the cell partition.
It is tangent to the discs centered on the vertices of the triangle which corresponds in the dual to its center.
Support discs may overlap each other but they do not overlap any disc of the packing.
}
\label{fig:FM_triangulation}
\end{figure}

A disc is said to be {\em surrounded} by discs of radii $r_1,\ldots,r_k$ if the sequence of its neighbors in the triangulation, ordered by angles and up to a cyclic permutation or a reversal, are discs of radii $r_1,\ldots, r_k$.
For example, in Fig.~\ref{fig:FM_triangulation}, the large central disc is surrounded by discs of radii $1,1,r,1,r,r,r$.

Introduced in \cite{FM58} (see also \cite{FT64}), FM-triangulations are also known as {\em additively weighted Delaunay triangulations}.
A triangle $T$ which appears in the FM-triangulation of some disc packing is said to be {\em feasible}.
The following property somehow extends the ``empty disc property'' of classic Delaunay triangulations:

\begin{proposition}[\cite{FM58}]
\label{prop:support_disc}
A triangle $T$ which connects the centers of three discs of a disc packing is feasible if and only if there exists a disc, called the {\em support disc} of $T$, which is interior disjoint from the discs of the packing and tangent to each of the three discs of $T$.
\end{proposition}

We shall also rely on the following property, which may be false in classic Delaunay triangulations when the ratio of disc radius is greater than $\sqrt{2}-1$:

\begin{proposition}[\cite{FM58}]
\label{prop:sector}
The disc sector delimited by any two edges of a feasible triangle $T$ never crosses the third edge of $T$.
\end{proposition}

A simple consequence we shall rely on (in particular in the computer program) is that the angles of feasible triangles cannot be too small:

\begin{lemma}
\label{lem:min_angles}
Let $T$ be a triangle in a FM-triangulation of a saturated packing by discs of radius $1$ and $r$.
Denote by  $A$, $B$ and $C$ the vertices of $T$ and by $x$, $y$ and $z$ the radii of the discs centered on these vertices.
Then, the angle $\widehat{A}$ in the vertex $A$ satisfies
$$
\sin\widehat{A}\geq\min\left(\frac{y}{x+2r+y},\frac{z}{x+2r+z}\right).
$$
\end{lemma}

\begin{proof}
Assume that the edge AB is shorter than AC.
On the one hand, the altitude of $T$ through B is at least $y$ because otherwise the disc sector defined by edges BA and BC would cross the edge AC.
On the other hand, the length of the edge AB is at most $x+2r+y$ because we can connect both A and B to the center of the support disc, which has radius at most $r$ since the packing is saturated.
This yields $\sin\hat{A}\geq\tfrac{y}{x+2r+y}$.
The same holds exchanging $y$ and $z$ if AC is shorter than AB, whence the claimed lower bound.
\end{proof}

Following \cite{Hep03}, we call a triangle whose discs are mutually adjacent a {\em tight triangle}  (Fig.~\ref{fig:tight}).
In particular, the FM-triangulation of a compact packing is made of tight triangles only and coincide with the contact graph.
Indeed, each hole in a compact disc packing is bounded by three mutually adjacent discs which thus define a tight triangle, and inserting in this hole a disc tangent to the three discs yields a support disc of this triangle which is thus feasible according to Prop.~\ref{prop:support_disc}.

\begin{figure}[hbt]
\centering
\includegraphics[width=\textwidth]{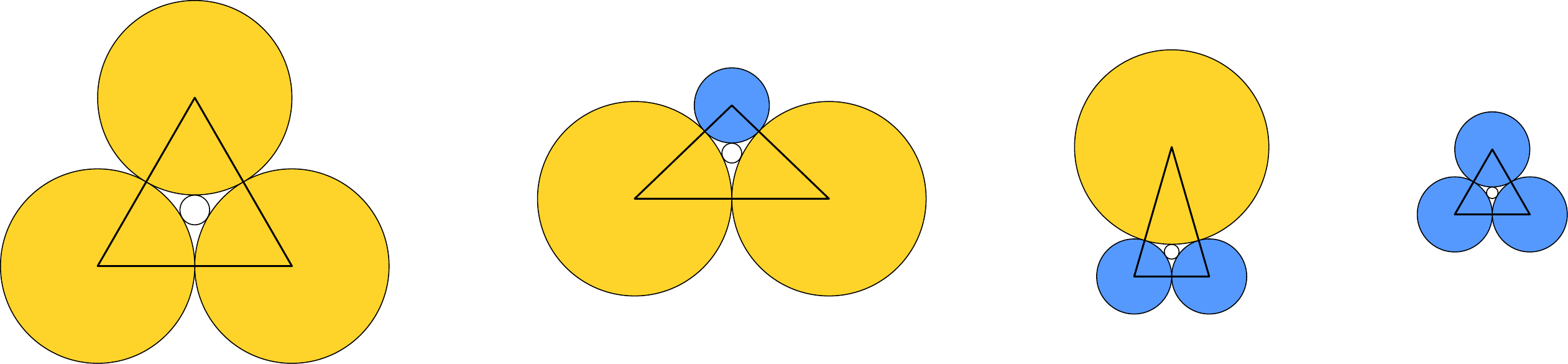}
\caption{The four types of tight triangles, with their support disc (in white).}
\label{fig:tight}
\end{figure}

Still following \cite{Hep03}, we call a triangle with a small disc tangent to both the two other discs, as well as to the line which passes through their centers, a {\em stretched triangle} (Fig.~\ref{fig:stretched}).
They are not feasible in a saturated packing because their support disc have radius $r$ (this would allow to add a small disc), but they can be arbitrarily approached near by feasible triangles.
We shall see in Section~\ref{sec:global_edge} that stretched triangles are dangerous because they can be as dense as tight triangles.
Indeed, two stretched triangles adjacent along their ``stretched edge'' (the edge tangent to one of the disc) can be recombined into two tight triangles by flipping this stretched edge.

\begin{figure}[hbt]
\centering
\includegraphics[width=\textwidth]{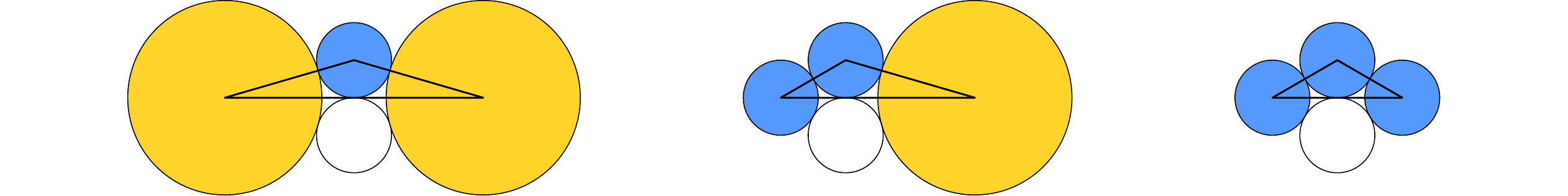}
\caption{The three types of stretched triangles, with their support discs.}
\label{fig:stretched}
\end{figure}

\section{Global inequality for the vertex potential}
\label{sec:global_vertex}

We shall here define the vertex potential and show it satisfies Inequality~\eqref{eq:vertex}.

\subsection{Vertex potential in tight triangles}
\label{sec:vertex_pot_tight}

For the sake of simplicity, we search for a vertex potential which depends only on the radii of the disc in the vertex and the radii of the two other discs in the triangle.
We denote by $V_{abc}$ (or $V_{cba}$) the potential in the center of the disc of radius $b$ in a tight triangle with discs of radius $a$, $b$ and $c$.
There are thus $6$ quantities to be defined:
$$
V_{111},\quad V_{rrr},\quad V_{r1r},\quad V_{1rr},\quad V_{1r1},\quad V_{11r}.
$$
We are going to impose constraints that these potentials will have to satisfy.

First, we impose that the sum of the vertex potential in each tight triangle is equal to its emptiness.
This yields four equations on the $V_{abc}$'s:
$$
3V_{111}=E_{111},\quad
3V_{rrr}=E_{rrr},\quad
V_{r1r}+2V_{1rr}=E_{1rr},\quad
V_{1r1}+2V_{11r}=E_{11r},
$$
where $E_{abc}$ i   s the emptiness of a tight triangle with discs of radius $a$, $b$ and $c$.

Second, we impose that, around any vertex $v$ of the target packing, Inequality~\eqref{eq:vertex} is an equality.
This thus yields an equation for each configuration around a vertex of the target packings.
Remarkably, there is only one equation for each radius of disc in each target packing, except for $\mathcal{P}_5$ where a small disc can be surrounded in two different ways.
This latter case is a bit specific and will be dealt with in Sec.~\ref{sec:c5}.
Tab.~\ref{tab:coronas_eq} lists these equations.

\begin{table}[hbt]
\centering
\begin{tabular}{c|cc}
$i$ & ~\hfill Small disc \hfill~ & ~\hfill Large disc \hfill~\\
\hline
$1$ & $3V_{1r1}+2V_{1rr}$ & $6V_{11r}+V_{r1r}$\\
$2$ & $V_{rrr}+2V_{1r1}+2V_{1rr}$ & $4V_{11r}+2V_{111}+V_{r1r}$\\
$3$ & $4V_{1rr}+V_{1r1}$ & $4V_{r1r}+4V_{11r}$\\
$4$ & $4V_{1r1}$ & $8V_{11r}$\\
$6$ & $V_{rrr}+4V_{1rr}$ & $12V_{r1r}$\\
$7$ & $2V_{1rr}+2V_{1r1}$ & $8V_{11r}+2V_{r1r}$\\
$8$ & $3V_{1r1}$ & $12V_{11r}$\\
$9$ & $V_{rrr}+2V_{1rr}+V_{1r1}$ & $12V_{11r}+6V_{r1r}$\\
\end{tabular}
\caption{
Quantities that must be zero to have equality in Inequality~\eqref{eq:vertex} around the center of a small or a large disc in the target packings depicted in Fig.~\ref{fig:targets}.}
\label{tab:coronas_eq}
\end{table}

We thus have six equations for each target packing (one in each of the four tight triangles and one around each of the two discs).
They are actually not independent because the sum of the emptiness of tight triangles over the fundamental domain of each target packing is equal to zero.
There is thus still one degree of freedom.
We arbitrarily set $V_{r1r}:=0$, except for $\mathcal{P}_6$ where we set $V_{1r1}:=0$ because $V_{r1r}=0$ is already enforced around a large disc.
A computation (performed in the joined program) shows that, in each case, these $6$ equations are independent.
All the $V_{abc}$'s are thus now fully determined.
When such potentials are used in a computer calculation, an interval containing the exact value will be used.
These numerical values being ugly and quite numerous ($6\times 9$), we do not list them here (they are moreover quite simple to compute from the above equations).

\subsection{Vertex potential in any triangle}
\label{sec:vertex_pot}

We shall now define the vertex potential in {\em any} triangle.
The idea is to modify the potential of a tight triangle depending on how much the triangle itself is deformed.
Given a triangle $T$ in a FM-triangulation of a disc packing, we denote by $T^*$ the tight triangle obtained in contracting the edges until the three discs become mutually tangent (such a triangle is always defined because if $r_a$, $r_b$ and $r_c$ are the radii of the discs, then the edge lengths are $r_a+r_b$, $r_b+r_c$ and $r_a+r_c$ and each of these length is greater than the sum of the two other ones).

\begin{definition}
\label{def:vertex_pot}
Let $v$ be a vertex of a triangle $T$.
Let $q$ be the radius of the disc of center $v$ and $x$ and $y$ the radii of the two other discs of $T$.
The {\em vertex potential} $U_v(T)$ of $v$ is defined by
\[
U_v(T):=V_{xqy}+m_q|\widehat{v}(T)-\widehat{v}(T^*)|,
\]
where $m_q\geq 0$ depends only on $q$, and $\widehat{v}(T)$ and $\widehat{v}(T^*)$ denote the angle in $v$ in $T$ and $T^*$.
\end{definition}

In particular, $U_v(T^*)=V_{xqy}$.
The constant $m_q$ controls the ``deviation'' in term of the angle changes between $T$ and $T^*$.
The point is to fix it so that the inequality~\eqref{eq:vertex} holds:

\begin{proposition}
\label{prop:vertex_pot}
Let $i\neq 5$ and $v$ be a vertex of an FM-triangulation of a saturated packing by discs of radius $1$ and $r_i$.
Then, the sum of the vertex potentials of the triangles containing $v$ is nonnegative provided that $m_1$ and $m_r$ are bounded from below by the values given in Tab.~\ref{tab:m}.
\end{proposition}

\begin{table}[hbt]
\centering
\begin{tabular}{c|cc}
$i$ & ~\hfill $m_1$ \hfill~ & ~\hfill $m_r$ \hfill~\\
\hline
$1$ & $0$ & $0.0005$\\
$2$ & $0.16$ & $0.087$\\
$3$ & $0$ & $0.00028$\\
$4$ & $0$ & $0.0021$\\
$6$ & $0.0091$ & $0.0021$\\
$7$ & $0$ & $0.0010$\\
$8$ & $0$ & $0.0020$\\
$9$ & $0$ & $0.002058$\\
\end{tabular}
\caption{
Lower bounds on $m_1$ and $m_r$ which ensure the vertex inequality~\eqref{eq:vertex} for any packing by discs of radius $1$ and $r_i$.
}
\label{tab:m}
\end{table}

\begin{proof}
Let $v$ be a vertex of an FM-triangulation $\mathcal{T}$ of a saturated packing by discs of radius $1$ and $r$.
Let $q$ denote the radius of the disc of center $v$.
Let $T_1,\ldots,T_k$ be the triangles of $\mathcal{T}$ which contain $v$, ordered clockwise around $v$.
We have:
\begin{eqnarray*}
\sum_{j=1}^k U_v(T_j)
&=&\sum_{j=1}^{k} U_v(T^*_j)+m_q\sum_{j=1}^k|\widehat{v}(T_j)-\widehat{v}(T^*_j)|\\
&\geq& \sum_{j=1}^k U_v(T^*_j)+m_q\left|\sum_{j=1}^k\widehat{v}(T_j)-\sum_{j=1}^k\widehat{v}(T^*_j)\right|.
\end{eqnarray*}
Since the $T_j$'s surround $v$, $\sum_j\widehat{v}(T_j)=2\pi$.
If the coefficient of $m_q$ is nonzero, then the inequality~\eqref{eq:vertex} is thus satisfied in $v$ as soon as
\[
m_q\geq -\frac{\sum_j U(T^*_j)}{\left|2\pi-\sum_j\widehat{v}(T^*_j)\right|}.
\]
This lower bound depends only on the radii and order of the discs centered on the neighbors of $v$.
There is only finitely many cases for each value of $k$, and the lower bounds on angles of Lemma~\ref{lem:min_angles} ensure that there is finitely many values of $k$ (the largest one is $k=80$, reached for $i=9$ when there are only small discs around a large one\footnote{We can actually reduce further the number of cases to consider by bounding from below the angle of a triangle depending on the discs of this triangles. It is however only useful to speed up the search, because the cases that give the lower bound on $m_q$ correspond to rather small values of $k$.}).
We can thus perform an exhaustive search on a computer to find a lower bound which holds for any $v$.
We performed this exhaustive search in the function \verb+smallest_m+ in \verb+binary.sage+).
To conclude, we also have to consider the case where $m_q$ has a zero coefficient.
This happens when the sum of the angles $\widehat{v}(T^*_j)$ is equal to $2\pi$.
We check this during the previous exhaustive search: if the computation yields for the coefficient of $m_q$ an interval which contains zero, then we check whether $\sum_j U(T^*_j)\geq 0$.
The computation shows that this always holds, except when $v$ is surrounded in the same way as in the target packing.
In this latter case, we get an interval which contains zero: this is the way we defined the vertex potential in tight triangles (namely, to satisfy the equations in Tab.~\ref{tab:coronas_eq}) which theoretically ensures that the exact value is zero.
\end{proof}

\subsection{The case $\mathcal{P}_5$}
\label{sec:c5}
 
We cannot proceed exactly the same way for $\mathcal{P}_5$, because equality in Inequality~\eqref{eq:vertex} around the small disc surrounded by six other small discs would yield $V_{rrr}=0$.
Since $E_{rrr}=3V_{rrr}$ (see Subsec.~\ref{sec:vertex_pot_tight}), this would yield $E_{rrr}=0$.
But since the density $\delta_5$ is larger than the density $\tfrac{\pi}{2\sqrt{3}}$ of the hexagonal compact packing, the emptiness $E_{rrr}$ must be positive.
To get around this problem, the potential in a vertex of a triangle with three small discs will depend on its neighborhood.

In an FM-triangulation of a packing by discs of radius $1$ and $r=r_5$, a small disc surrounded by two large discs and three small ones (in this order up to a cyclic permutation) is said to be {\em singular}.
The other small discs are said to be {\em regular}.
In particular, each regular small disc in the target packing $\mathcal{P}_5$ is surrounded by $6$ singular discs (Fig.~\ref{fig:targets}).
We shall rely on the following simple lemma:

\begin{lemma}
\label{lem:c5}
In an FM-triangulation of a packing by discs of radius $1$ and $r$, there is at most two singular discs in a triangle with three discs of radius $r$.
\end{lemma}

\begin{proof}
Assume that there is a triangle with three singular small discs (shaded triangle in Fig.~\ref{fig:lemma_c5}) and let us get a contradiction.
Consider a first disc (in dark grey in Fig.~\ref{fig:lemma_c5}): since it is singular it is surrounded by discs of radius 1,1,r,r,r (in this order), and since it already has two neighbors of radius r, all its neighbors are uniquely defined up to the orientation.
Consider now a second disc, say the one with known neighbors 1,r,r (in light grey in Fig.~\ref{fig:lemma_c5}): since it is singular, it is surrounded by discs of radius 1,1,r,r,r, and the known neighbors leave only one way to arrange the two unknown neighbors.
This yields known neighbors r,r,r,r for the third disc of the triangle, which is incompatible with being singular.
\end{proof}

\begin{figure}[hbt]
\centering
\includegraphics[width=0.4\textwidth]{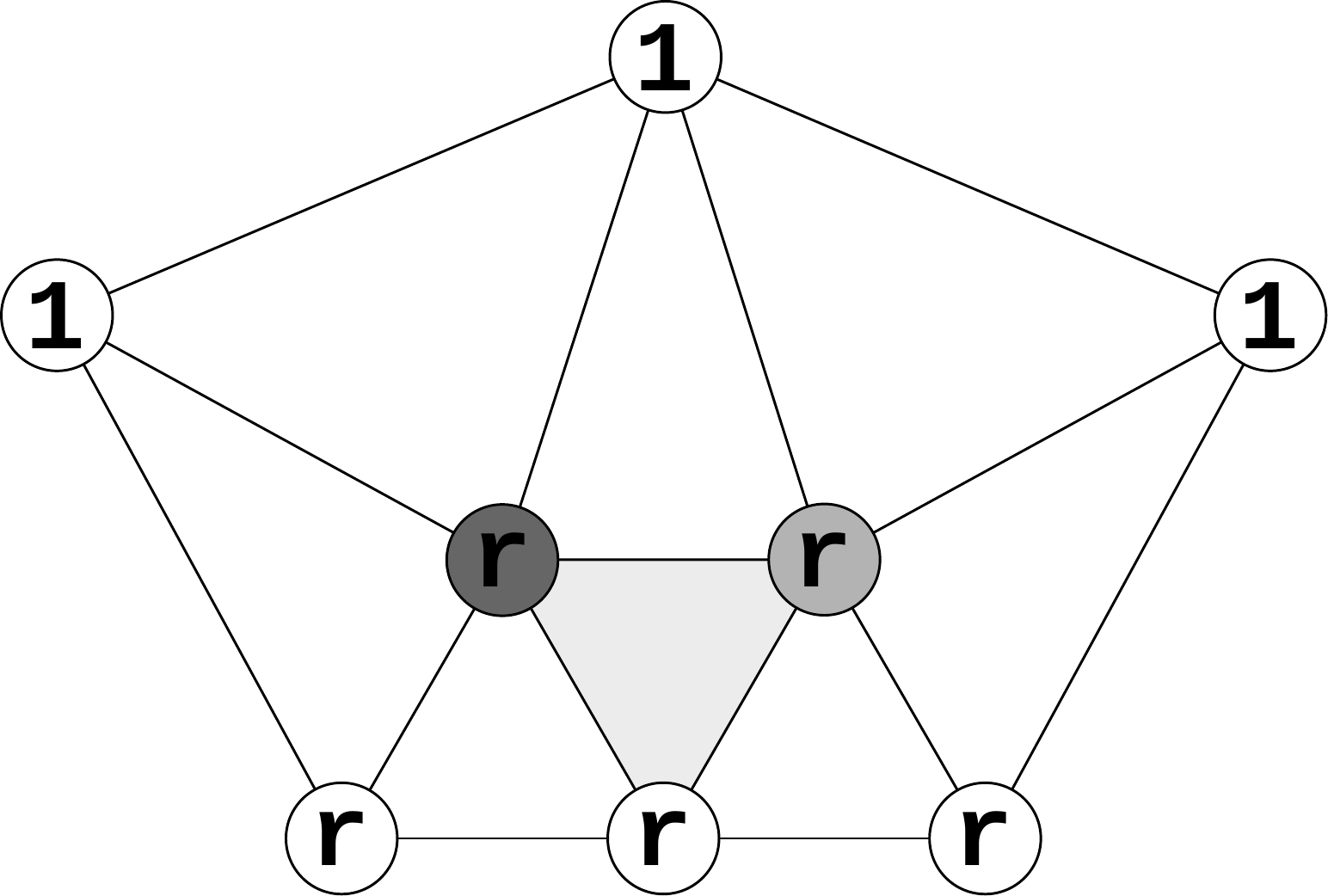}
\caption{
A triangle cannot have three singular discs (Lem.~\ref{lem:c5}).
}
\label{fig:lemma_c5}
\end{figure}

In a tight triangle with small discs, we denote respectively by $V'_{rrr}$ and $V_{rrr}$ the potentials of singular and regular vertices.
We set
$$
V'_{rrr}:=\tfrac{1}{2}E_{rrr}.
$$
The above lemma ensures that the sum over the singular vertices of the triangle is at most $E_{rrr}$.
The remaining potential (to sum up to $E_{rrr}$ on the triangle) is shared equally among the $k\geq 1$ regular vertices.
The value of $V_{rrr}$ thus depends on the number of singular vertices in the triangle: it can be $0$, $\tfrac{1}{4}E_{rrr}$ or $\tfrac{1}{3}E_{rrr}$.
In particular, it is always nonnegative.
While checking Inequality~\eqref{eq:vertex} around regular vertices, we shall only assume $V_{rrr}\geq 0$.
We can further proceed as for the other cases to define the $V_{abc}$'s, with $V'_{rrr}$ instead of $V_{rrr}$ and considering only singular vertices.
Tab.~\ref{tab:coronas_eq_c5} completes Tab.~\ref{tab:coronas_eq}.

\begin{table}[hbt]
\centering
\begin{tabular}{c|cc}
$i$ & Small disc (singular) & Large disc\\
\hline
$5$ & $V_{1r1}+2V_{1rr}+2V'_{rrr}$ & $6V_{11r}+3V_{r1r}$\\
\end{tabular}
\caption{
Quantities that must be zero to have equality in Inequality~\eqref{eq:vertex} around the center of a small or a large disc in in $\mathcal{P}_5$.}
\label{tab:coronas_eq_c5}
\end{table}

We can then extend vertex potentials beyond tight triangles exactly as in Definition~\ref{def:vertex_pot}, since the regular or singular character of a small disc is defined for any triangle.
Tab.~\ref{tab:m_c5} completes Tab.~\ref{tab:m} to extend Proposition~\ref{prop:vertex_pot}, which is proven in the same way, with the only difference being that in the exhaustive search through possible configurations around a vertex $v$, we simply use $V_{rrr}\geq 0$ if $v$ is not singular (since knowing only the neighbors of $v$ not always suffice to determine which of them are singular or regular).

\begin{table}[hbt]
\centering
\begin{tabular}{c|cc}
$i$ & ~\hfill $m_1$ \hfill~ & ~\hfill $m_r$ \hfill~\\
\hline
$5$ & $0$ & $0.0473$\\
\end{tabular}
\caption{
Lower bounds on $m_1$ and $m_r$ which ensure the vertex inequality~\eqref{eq:vertex} for any packing by discs of radius $1$ and $r_5$.
}
\label{tab:m_c5}
\end{table}

\section{Capping the potential}
\label{sec:capping}

In Subsec.~\ref{sec:vertex_pot_tight}, we fixed the vertex potentials in order to have $E(T)=U(T)$ on the tight triangles.
We then introduced, in Subsec.~\ref{sec:vertex_pot}, a deviation controlled by the quantities $m_1$ and $m_r$ to have Ineq.~\eqref{eq:vertex} around each vertex of any FM-triangulation of a saturated packing.
More precisely, we found lower bounds on $m_1$ and $m_r$: any largest values would only make this latter inequality even more true.
However, we shall keep in mind that we also have to eventually satisfy the local inequality~\eqref{eq:local}, {\em i.e.}, $U(T)\leq E(T)$ for any triangle.
With this in mind, it is best to fix $m_1$ and $m_r$ as small as possible so as to minimize $U$.
We can actually make $U$ even smaller as follows.

\begin{proposition}
Assume that, for any tight triangle $T^*$ with a disc of radius $q$ in $v$, one has
\begin{equation}
\label{eq:mq_for_capping}
m_q\geq -\frac{U_v(T^*)}{\widehat{v}(T^*)}.
\end{equation}
Then the vertex inequality~\eqref{eq:vertex} still holds if we cap the vertex potential $U_v(T)$ by
$$
z_q:=-2\pi \min \tfrac{U_v(T^*)}{\widehat{v}(T^*)},
$$
where the minimum is over the tight triangles $T^*$ with a disc of radius $q$ in $v$.
\end{proposition}

\begin{proof}
Note that $z_q>0$ since at least one of the $U_v(T^*)$'s is negative in order to have equality in Ineq.~\eqref{eq:vertex} for the target packing.
We shall show the following lower bound on the vertex potential {\em per radian}:
$$
\frac{U_v(T)}{\widehat{v}(T)}\geq \min\left(0,\frac{U_v(T^*)}{\widehat{v}(T^*)}\right).
$$
The definition of $z_q$ then ensures that, as soon as a vertex potential $U_v(T)$ is larger than $z_q$, the vertex potentials of all the other triangles sharing $v$ are ``not enough negative to be dangerous'', namely, the sum of all the vertex potentials around $v$ is nonnegative, that is, Ineq.~\eqref{eq:vertex} holds.

Let us prove the claimed lower bound per radian.
It is trivial if $U_v(T^*)\geq 0$.
Assume $U_v(T^*)\leq 0$.
If $\widehat{v}(T)\geq\widehat{v}(T^*)$, then
$$
\frac{U_v(T)}{\widehat{v}(T)}
=\frac{U_v(T^*)+m_q(\widehat{v}(T)-\widehat{v}(T^*))}{\widehat{v}(T)}
\geq\frac{U_v(T^*)}{\widehat{v}(T)}
\geq\frac{U_v(T^*)}{\widehat{v}(T^*)}.
$$
If $\widehat{v}(T)\leq\widehat{v}(T^*)$, then
$$
\frac{U_v(T)}{\widehat{v}(T)}
=\frac{U_v(T^*)+m_q(\widehat{v}(T^*)-\widehat{v}(T))}{\widehat{v}(T)}
=\frac{U_v(T^*)+m_q\widehat{v}(T^*)}{\widehat{v}(T)}-1.
$$
Since $U_v(T^*)+m_q\widehat{v}(T^*)\geq 0$ by Ineq.~\eqref{eq:mq_for_capping}, the above quantity is a decreasing function of $\widehat{v}(T)$.
It is thus bounded from below by its value for $\widehat{v}(T)=\widehat{v}(T^*)$, {\em i.e.}, $U_v(T^*)/\widehat{v}(T^*)$.
\end{proof}

A computation shows that the values of $m_r$ given in Tab.~\ref{tab:m} and \ref{tab:m_c5} satisfy the condition~\eqref{eq:mq_for_capping}.
This also holds for $m_1$ in cases $c_2$ and $c_6$, but not in the other cases because $m_1$ is equal to zero.
However, the value required to satisfy the condition~\eqref{eq:mq_for_capping} is quite small: a computation shows that $10^{-14}$ suffices for all the cases.
We thus {\em modify} $m_1$ when its value in Tab.~\ref{tab:m} and \ref{tab:m_c5} is zero and increase it to $10^{-14}$.

Since only the negative $V_{abc}$'s play a role in the definition of $z_q$, the value of $V_{rrr}$ in the case $\mathcal{P}_5$, which can range from $0$ to $\tfrac{1}{2}E_{rrr}$ has no importance.
The values $Z_1$ and $Z_r$ listed in Tab.~\ref{tab:Z} bound from above the exact values $z_1$ and $z_r$.

\begin{table}[hbt]
\centering
\begin{tabular}{c|cc}
$i$ & $Z_1$ & $Z_r$\\
\hline
$1$ & $7.5\times 10^{-15}$ & $0.00023$\\
$2$ & $0.010$ & $0.0045$\\
$3$ & $1.7\times 10^{-14}$ & $0.00023$\\
$4$ & $4.89\times 10^{-15}$ & $0.00095$\\
$5$ & $1.04\times 10^{-16}$ & $0.0076$\\
$6$ & $0.00124$ & $0.0015$\\
$7$ & $8.92\times 10^{-15}$ & $0.0011$\\
$8$ & $8.0\times 10^{-15}$ & $0.0012$\\
$9$ & $2.327\times 10^{-14}$ & $0.0008032$
\end{tabular}
\caption{Values of $Z_1$ and $Z_r$ used to cap the vertex potentials.}
\label{tab:Z}
\end{table}

\section{Global inequality for the edge potential}
\label{sec:global_edge}

A few randomized trials suggest that the vertex potential satisfies the local inequality~\eqref{eq:local} for triangles which are not too far from tight triangles.
It however fails near stretched triangles, because the emptiness can become quite small.
A typical situation is depicted in Fig.~\ref{fig:plot}.
The edge potential aims to fix this problem.
The idea is that when a triangle $T$ becomes stretched, its support disc overlaps an adjacent triangle $T'$, imposing a void in $T'$ which increases the emptiness $E(T')$ and may counterbalance the decrease of $E(T)$.
We shall come back to this in Section~\ref{sec:local}.
Here, we define the edge potential and prove that it satisfies Inequality~\eqref{eq:edge}.

\begin{definition}
\label{def:edge_pot}
Let $e$ be an edge of a triangle $T$.
Denote by $d_e(T)$ the signed distance of the center $X$ of the support disc of $T$ to the edge $e$, which is positive if $T$ and $X$ are both on the same side of $e$, or negative otherwise.
The {\em edge potential} $U_e(T)$ of $e$ is defined by
\[
U_e(T):=\left\{\begin{array}{l}
\textrm{$0$ if the edge $e$ is shorter than $l_e$,}\\
\textrm{$q_e\cdot d_e(T)$ otherwise,}
\end{array}\right.
\]
where $l_e$ and $q_e$ are positive constants which depend only on the sizes of the discs centered on the endpoints of the edge $e$.
\end{definition}

\begin{figure}[hbt]
\centering
\includegraphics[width=\textwidth]{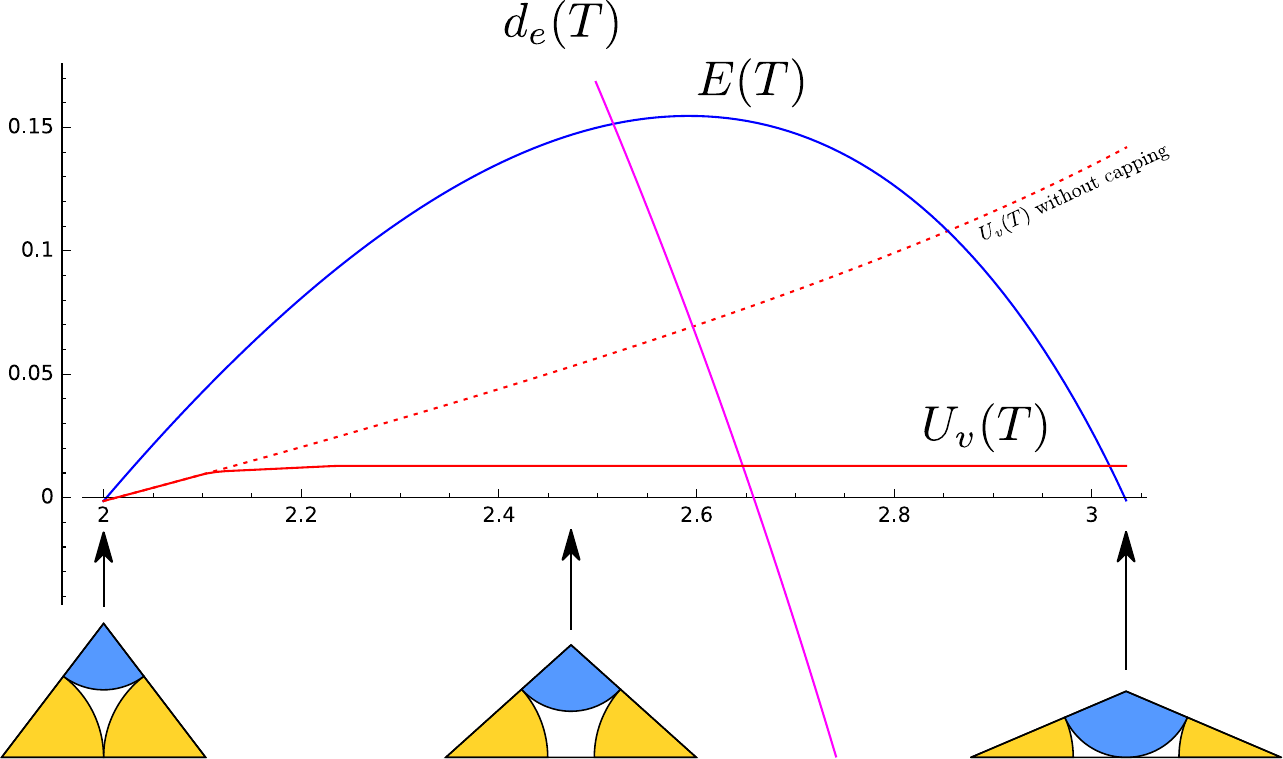}
\caption{
Starting from a tight triangle with two large disc and a small one (bottom left), the edge $e$  between the two large discs is elongated until we get a stretched triangle (bottom right).
The corresponding variations of the emptiness $E$, vertex potential $U_v$ and signed distance $d_e(T)$ are depicted (top).
For quasi-stretched triangle, the local inequality $E(T)\geq U_v(T)$ fails.
}
\label{fig:plot}
\end{figure}

The constant $l_e$ is the threshold below which $d_e$ has an effect and the coefficient $q_e$ controls the intensity of this effect.
As explained at the beginning of this section, the idea is that this edge potential will only come into play for triangles that are ``quite stretched''.
This is why the value of the threshold $l_e$ will in practice be much higher than the minimal length of the edge $e$ (obtained when the two disks are in contact), which can be seen in Table~\ref{tab:LQ}.
Nevertheless, whatever the value of this threshold, the edge potential always satisfies Inequality ~\eqref{eq:edge}:

\begin{proposition}
\label{prop:edge_pot}
If $e$ is an edge of an FM-triangulation of a disc packing, then the sum of the edge potentials of the two triangles containing $e$ is nonnegative.
\end{proposition}

\begin{proof}
Consider an edge $e$ shared by two triangles $T$ and $T'$ of an FM-triangulation.
We claim that $d_e(T)+d_e(T')\geq 0$.
If each triangle and the center of its support disc are on the same side of $e$, then it holds because both $d_e(T)$ and $d_e(T')$ are nonnegative.
Assume $d_e(T)\leq 0$, {\em i.e}, $T$ and the center of its support disc are on either side of $e$.
Denote by $A$ and $B$ the endpoints of $e$ and by $a$ and $b$ the radii of the discs of center $A$ and $B$ (Fig.~\ref{fig:edge_pot}).
The centers of the discs tangent to both discs of center $A$ and $B$ and radius $a$ and $b$ are the points $M$ such that $\overline{AM}-a=\overline{BM}-b$, {\em i.e.}, a branch of a hyperbola of foci $A$ and $B$.
This includes the centers $X$ and $X'$ of the support discs of $T$ and $T'$.
In order to be tangent to the third disc of $T'$, the support disc of $T'$ must have a center $X'$ farther than $X$ from the focal axis.
Since the distances of $X'$ and $X$ to this axis are $-d_e(T)$ and $d_e(T')$, this indeed yields $d_e(T)+d_e(T')\geq 0$.
This proves $U_e(T)+U_e(T')\geq 0$ if $e$ has length at least $l_e$.
If $e$ is shorter than $l_e$, both $U_e(T)$ and $U_e(T')$ are zero and their sum is thus also nonnegative.
\end{proof}

\begin{figure}[hbt]
\centering
\includegraphics[width=\textwidth]{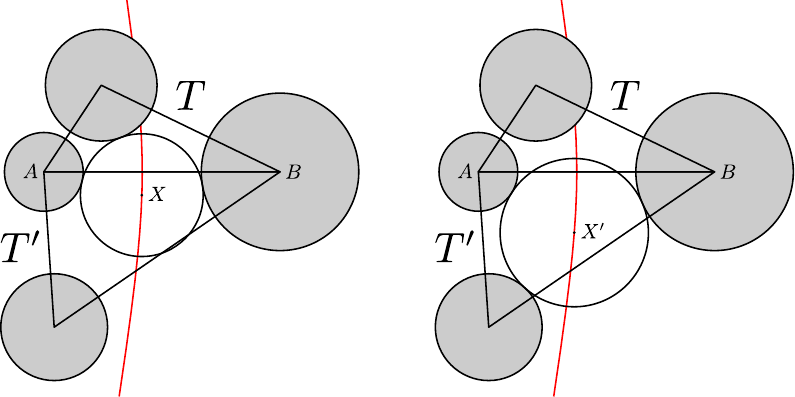}
\caption{
On the left, the support disc of the top triangle $T$ is centered at $X$.
On the right, the support disc of the bottom triangle $T'$ is centered at $X'$.
the hyperbola branch indicates the locations of the centers of the disks tangent to the two disks centered at A and B.
}
\label{fig:edge_pot}
\end{figure}

\section{Local inequality for $\varepsilon$-tight triangles}
\label{sec:local_tight}

We prove the local inequality~\eqref{eq:local} in a neighborhood of tight triangles.
A triangle is said to be $\varepsilon$-tight if its discs are pairwise at distance at most $\varepsilon$, that is, each edge between discs of radii $x$ and $y$ has length between $x+y$ and $x+y+\varepsilon$.

Let $T^*$ be a tight triangle with edge length $x_1$, $x_2$ and $x_3$ and denote by $\mathcal{T}_\varepsilon$ the set of $\varepsilon$-tight triangles with the same disc radii as $T^*$.
On the one hand, the variation $\Delta E$ of the emptiness $E$ between $T^*$ and any triangle in $\mathcal{T}_\varepsilon$ satisfies
\[
\Delta E\geq \sum_{1\leq i\leq 3} \min_{\mathcal{T}_\varepsilon}\frac{\partial E}{\partial x_i}\Delta x_i.
\]
One the other hand, the variation $\Delta U$ of the potential $U$ between $T^*$ and any triangle in $\mathcal{T}_\varepsilon$ satisfies:
\[
\Delta U\leq\sum_{1\leq i\leq 3} \max_{\mathcal{T}_\varepsilon}\frac{\partial U}{\partial x_i}\Delta x_i.
\]
Since $E(T^*)=U(T^*)$ by definition of the vertex potential (Subsection~\ref{sec:vertex_pot_tight}), the local inequality $E(T)\geq U(T)$ holds over $\mathcal{T}_\varepsilon$ for any $\varepsilon$ such that
\[
\min_{\mathcal{T}_\varepsilon}\frac{\partial E}{\partial x_i}\geq\max_{\mathcal{T}_\varepsilon}\frac{\partial U}{\partial x_i}.
\]
To compute the derivatives of $U$, it is convenient to assume that the threshold $l_e$ above which the edge potential comes into play is larger than $x+y+\varepsilon$, where $x$ and $y$ are the radii of the discs connected by the edge $e$.
Indeed, the potential $U$ is then equal to the vertex potential $U_v$ over $\mathcal{T}_\varepsilon$, and the computation is much simpler.
This assumption is largely satisfied in practice, as can be seen from the values given in Tables~\ref{tab:eps} and \ref{tab:LQ}\footnote{For example, for $x=y=1$ in the $i=1$ case, we have $x+y+\varepsilon=2.0078$ and $l_e=2.70$.}.
We computed  the formulas of the derivatives of $E$ and $U$ with SageMath\footnote{It can be done by hand since it mainly amounts to use the cosine theorem to express the angle of a triangle as a function of its edge length, but we are not particularly interested in the formulas.}.
We then once again use interval arithmetic to compute the extreme values over $\mathcal{T}_\varepsilon$: each variable $x_i$ is replaced by the interval $[r_j+r_k, r_j+r_k+\varepsilon]$, where $r_j$ and $r_k$ denote the radii of the discs centered on the endpoints of the edge of length $x_i$.
The computation yields:

\begin{proposition}
\label{prop:local_tight}
Let $i\in\{1,\ldots,9\}$.
Take for $m_1$ and $m_r$ the lower bound given in Tab.~\ref{tab:m} or \ref{tab:m_c5}.
Take for $Z_1$ and $Z_r$ the values given in Tab.~\ref{tab:Z}.
Then, the local inequality $E(T)\geq U(T)$ holds for any $\varepsilon$-tight triangle of an FM-triangulation of a saturated packing by discs of radius $1$ and $r_i$ provided that $\varepsilon$ satisfies the upper bound given in Tab.~\ref{tab:eps}.
\end{proposition}

\begin{table}[hbt]
\centering
\begin{tabular}{c|c}
$i$ & $\varepsilon$ \\
\hline
$1$ & $0.078$\\
$2$ & $0.019$\\
$3$ & $0.060$\\
$4$ & $0.038$\\
$5$ & $0.0126$\\
$6$ & $0.026$\\
$7$ & $0.0186$\\
$8$ & $0.0048$\\
$9$ & $0.001717$\\
\end{tabular}
\caption{
Upper bounds on $\varepsilon$ which ensure Inequality~\eqref{eq:local} for $\varepsilon$-tight triangles.
}
\label{tab:eps}
\end{table}

\section{Local inequality for all the triangles}
\label{sec:local}

We here explicitly define an edge potential such that the local inequality~\eqref{eq:local} holds for any feasible triangle.
The following lemma completes the lemma~\ref{lem:min_angles} in order to eliminate as many as posible nonfeasible triangles during the computer check:

\begin{lemma}
\label{lem:non_feasible}
If a triangle $T$ appears in an FM-triangulation of a saturated packing by discs of radius $1$ and $r$, then
\begin{itemize}
\item the radius of its support disc is less than $r$;
\item its area is at least $\tfrac{1}{2}\pi r^2$;
\item for any vertex A of $T$, the altitude of $T$ through A is at least $r$.
\end{itemize}
\end{lemma}

\begin{proof}
\begin{itemize}
\item If the support disc has radius $r$ or more, then we can add a small disc in the packing, in contradiction with the saturation hypothesis.
\item According to Prop.~\ref{prop:sector} (page \pageref{prop:sector}), the sectors defined by triangle edges of the discs centered in the triangle vertices are included in the triangle.
Their total area is at least half the area of a small disc, whence the claimed lower bound.
\item If the altitude through vertex A is less than $r$, then the sector of the disc centered in A crosses the line going through the two opposite vertices B and C.
It cannot crosses it between B and C (Prop.~\ref{prop:sector}).
Then, the altitude is smaller through the nearest vertex to A, say B, than through A itself.
The sector of the disc centered in B thus crosses the segment AC (Fig.~\ref{fig:non_feasible}).
This contradicts Prop.~\ref{prop:sector}
\end{itemize}
\end{proof}

\begin{figure}[hbt]
\centering
\includegraphics[width=0.6\textwidth]{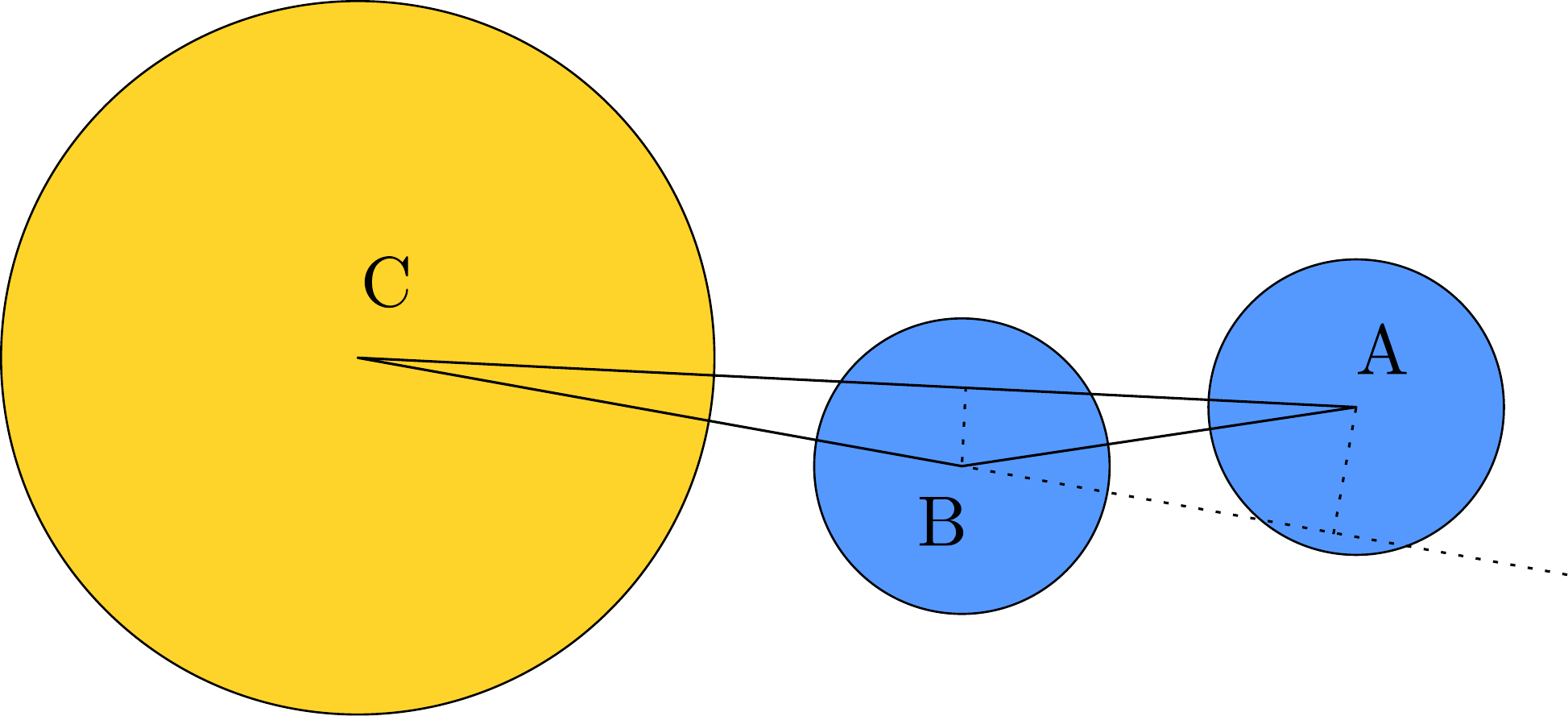}
\caption{A triangle with a vertex of altitude less than $r$ cannot be feasible.}
\label{fig:non_feasible}
\end{figure}

\begin{proposition}
\label{prop:local}
Let $i\in\{1,\ldots,9\}$.
Take for $m_1$ and $m_r$ the lower bound given in Tab.~\ref{tab:m} or \ref{tab:m_c5}.
Take for $Z_1$ and $Z_r$ the values given in Tab.~\ref{tab:Z}.
Take for $\varepsilon$ the value given in Tab.~\ref{tab:eps}.
Take for $l_e$ and $q_e$ the values given in Tab.~\ref{tab:LQ}.
Then, the local inequality $E(T)\geq U(T)$ holds for any triangle of an FM-triangulation of a saturated packing by discs of radius $1$ and $r_i$.
\end{proposition}

\begin{table}[hbt]
\centering
\begin{tabular}{c|cccccc}
$i$ & $l_{11}$ & $q_{11}$ & $l_{1r}$ & $q_{1r}$ & $l_{rr}$ & $q_{rr}$\\
\hline
$1$ & $2.70$ & $0.10$ & $2.30$ & $0.10$ & $1.90$ & $0.10$\\
$2$ & $2.60$ & $0.20$ & $2.10$ & $0.20$ & $1.60$ & $0.20$\\
$3$ & $2.60$ & $0.10$ & $2.20$ & $0.20$ & $1.65$ & $0.10$\\
$4$ & $2.50$ & $0.15$ & $1.80$ & $0.20$ & $1.20$ & $0.20$\\
$5$ & $2.40$ & $0.05$ & $1.80$ & $0.05$ & $1.10$ & $0.07$\\
$6$ & $2.50$ & $0.20$ & $1.75$ & $0.20$ & $1.00$ & $0.20$\\
$7$ & $2.40$ & $0.00$ & $1.60$ & $0.05$ & $0.80$ & $0.10$\\
$8$ & $2.24$ & $0.020$ & $1.33$ & $0.015$ & $0.44$ & $0.020$\\
$9$ & $2.17$ & $0.020$ & $1.217$ & $0.015$ & $0.2857$ & $0.020$
\end{tabular}
\caption{Values $l_e$ and $q_e$ which define the edge potentials, where $l_{xy}$ and $q_{xy}$ stands for an edge $e$ whose discs have radii $x$ and $y$.
More details on how these values have been chosen are given in Appendix~\ref{sec:LQ}.
}
\label{tab:LQ}
\end{table}

\begin{proof}
We shall check the inequality over all the possible triangles with the computer.
For $x\leq y\leq z$ in $\{1,r\}$, any triangle with discs of radius $x$, $y$ and $z$ which appear in an FM-triangulation of a saturated packing has edge length in the compact set
\[
[x+y,x+y+2r]\times[x+z,x+z+2r]\times [y+z,y+z+2r].
\]
Indeed, its support disc has radius at most $r$ (saturation hypothesis) so that the center of a disc of radius $q$ is at distance at most $q+r$ from the center of the support disc.
We can thus compute $E(T)$ and $U(T)$ using these intervals for the edge lengths of $T$.

Of course, since these intervals are quite large, we get for $E(T)$ and $U(T)$ large overlapping intervals which do not allow to conclude whether $E(T)\geq U(T)$ or not.
We use dichotomy: while the intervals are too large to conclude, we halve them and check recursively on each of the $2^3$ resulting compacts whether $E(T)\geq U(T)$ or not.
If we get $E(T)\geq U(T)$ at some step, we stop the recursion.
If we get $E(T)<U(T)$ at some step, we throw an error: the local inequality is not satisfied!

At each step, we also check whether Lemma~\ref{lem:non_feasible} ensures that the triangle is not feasible, in which case we eliminate it and stop the recursion (the way we compute the radius of the support disc is detailed in Appendix~\ref{sec:support_disc}).
Last, we also stop the recursion if we get an $\varepsilon$-tight triangle at some step, that is, if we get a subset of the compact
\[
[x+y,x+y+\varepsilon]\times[x+z,x+z+\varepsilon]\times [y+z,y+z+\varepsilon].
\]
Indeed, the local inequality is then already ensured by Prop.~\ref{prop:local_tight}.
This point is crucial and explains why we focused on $\varepsilon$-tight triangles in Section~\ref{sec:local_tight}.
Since $E(T)=U(T)$ for any tight triangle $T$, any non-empty interior compact which contains the point $(x+y,x+z,y+z)$
, no matter how small it is, yields for $E(T)$ and $U(T)$ overlapping intervals which do not allow to decide whether $E(T)>U(T)$ or not: the recursion would last forever!

For each $i$, the whole process terminates without throwing any error.
On our computer with our (noncompilated) python implementation, cases $1$--$7$ are checked in a few dozen seconds each, case $8$ in around 25 min. and case $9$ in 3 h. 30 min.
Tab.~\ref{tab:stats} gives some statistics on the number of checked triangles.
This proves the proposition.
\end{proof}

\begin{table}[hbt]
\centering
\begin{tabular}{c|cccc}
$i$ & 111 & 11r & 1rr & rrr \\
\hline
$1$ & $1940$ & $5559$ & $7547$ & $6000$\\
$2$ & $2633$ & $15961$ & $25201$ & $21211$\\
$3$ & $1842$ & $8443$ & $8261$ & $5405$\\
$4$ & $1415$ & $13406$ & $20357$ & $5755$\\
$5$ & $1128$ & $25691$ & $64534$ & $36786$\\
$6$ & $1275$ & $12818$ & $25943$ & $7393$\\
$7$ & $778$ & $22093$ & $62805$ & $4859$\\
$8$ & $232$ & $180391$ & $2316371$ & $17305$\\
$9$ & $92$ & $535858$ & $19069436$ & $19622$\\
\end{tabular}
\caption{
Number of triangles of each type on which the local inequality~\eqref{eq:local} had to be checked.
The 1rr-triangles for $i=9$ are by far the hardest case.
}
\label{tab:stats}
\end{table}

\section{Proof of Theorem~\ref{th:main}}

To conclude, let us shortly recall from Section~\ref{sec:localization} how the proven results fit together to obtain Theorem~\ref{th:main}.

Proposition~\ref{prop:vertex_pot} (see also Subsection~\ref{sec:c5} for the case $5$) proves the vertex inequality~\eqref{eq:vertex}.
Proposition~\ref{prop:edge_pot} proves the edge inequality~\eqref{eq:edge}.
Together, these inequalities yield the global inequality~\eqref{eq:global}.

Proposition~\ref{prop:local_tight} proves the local inequality~\eqref{eq:local} for $\varepsilon$-tight triangles, that is, almost tight triangles.
Note that, for tight triangles, the inequality~\eqref{eq:local} is an equality, which is ensured by the defintion of vertex potentials in tight triangles (Subsection~\ref{sec:vertex_pot}).
Proposition~\ref{prop:local} proves the local inequality~\eqref{eq:local} for the remaining triangles.
It relies on a dichotomy approach which eventually terminates because $E-U$ is uniformly bounded away from zero over these triangles (this is why we considered separately $\varepsilon$-tight triangles).
Note that this is also here that the capping described in Section~\ref{sec:capping} plays a role.

The global inequality~\eqref{eq:global} and the local inequality~\eqref{eq:local} ensure that the density $\delta$ of any packing is less or equal than the density $\delta_i$ of the target packing, as explained in the beginning of Section~\ref{sec:localization}.
This proves Theorem~\ref{th:main}.

\appendix
\section{Radii and densities}
\label{sec:exact}

Tab.~\ref{tab:radii} and \ref{tab:densities} give minimal polynomials of the radii $r_i$'s and the reduced densities $\delta_i/\pi$.
The polynomials for the radii come from \cite{Ken06}.
To compute the reduced density, consider a fundamental domain depicted in Fig.~\ref{fig:targets}.
The total area covered by discs divided by $\pi$ is an algebraic number since the radii are algebraic.
The area of any tight triangle is algebraic since so are the radii -- hence the edge length.
The quotien is thus an algebraic number whose minimal polynomial is the one given in Tab.~\ref{tab:densities}.

\begin{table}[hbt]
{\small 
\begin{enumerate}
\item $x^4 - 10x^2 - 8x + 9$,
\item $x^8 - 8x^7 - 44x^6 - 232x^5 - 482x^4 - 24x^3 + 388x^2 - 120x + 9$,
\item $8x^3 + 3x^2 - 2x - 1$,
\item $x^2 + 2x - 1$,
\item $9x^4 - 12x^3 - 26x^2 - 12x + 9$,
\item $x^4 - 28x^3 - 10x^2 + 4x + 1$,
\item $2x^2 + 3x - 1$,
\item $3x^2 + 6x - 1$,
\item $x^2 - 10x + 1$.
\end{enumerate}
}
\caption{Minimal polynomial of the radius $r_i$.}
\label{tab:radii}
\end{table}

\begin{table}[hbt]
{\small 
\begin{enumerate}
\item $27x^4 + 112x^3 + 62x^2 + 72x - 29$,
\item $x^8 - 4590x^6 - 82440x^5 + 486999x^4 - 1938708x^3 + 2158839x^2 - 1312200x + 243081$,
\item $1024x^3 - 692x^2 + 448x - 97$,
\item $2x^2 - 4x + 1$
\item $944784x^4 - 3919104x^3 - 2191320x^2 - 1632960x + 757681$,
\item $144x^4 + 9216x^3 + 133224x^2 - 127104x + 25633$,
\item $4096x^4 + 2924x^2 - 289$,
\item $108x^2 + 288x - 97$,
\item $144x^4 - 4162200x^2 + 390625$.
\end{enumerate}
}
\caption{Minimal polynomial of the reduced density $\delta_i/\pi$.}
\label{tab:densities}
\end{table}

\section{Support disc}
\label{sec:support_disc}

Consider an FM-triangle with sides of length $a$, $b$ and $c$.
Denote by $\alpha$ (resp. $\beta$ and $\gamma$) the vertex opposite to the edge of length $a$ (resp. $b$ and $c$).
Denote by $r_a$ (resp. $r_b$ and $r_c$) the radius of the disc of center $\alpha$ (resp. $\beta$ and $\gamma$).
We here explain how to get a formula that allows to compute the radius $R$ of the support disc (which exists and is unique for an FM-triangle), even when edge lengths or disc radii are intervals.

Fix a Cartesian coordinate system with $\alpha=(0,0)$ and $\beta=(c,0)$.
Denote by $(u,v)$ the coordinates of $\gamma$, with $v>0$.
One has:
$$
u=\frac{b^2+c^2-a^2}{2c}
\quad\textrm{and}\quad
v=\sqrt{b^2-u^2}.
$$
Denote by $(x,y)$ the coordinates of the center of the support disc and by $R$ its radius.
The definition of the support disc yields three equations:
\begin{eqnarray*}
(R+r_a)^2&=&x^2+y^2,\\
(R+r_b)^2&=&(b-x)^2+y^2,\\
(R+r_c)^2&=&(u-x)^2+(v-y)^2.
\end{eqnarray*}
Subtracting the second equation from the first yields an expression in $R$ for $x$.
Then, subtracting the third equation from the first and replacing $x$ by its expression yields an expression in $R$ for $y$.
Last, replacing $x$ and $y$ by their expressions in the third equations yields a quadratic equation $AR^2+BR+C=0$, where $A$, $B$ and $C$ are complicated but explicit polynomials in $a$, $b$, $c$, $r_a$, $r_b$ and $r_c$ (they appear in the code of the function \verb+radius+ in \verb+binary.sage+).
The root we are interested in is the smallest positive one because we want the support disc to be interior-disjoint from the three discs\footnote{One can check that the discriminant of this quadratic polynomial is the product of the square of the area of the triangle and the terms $2(x-r_y-r_z)$ for each permutation $(x,y,z)$ of $(a,b,c)$, hence always non-negative.}.
Given exact values $a$, $b$, $c$, $r_a$, $r_b$ and $r_c$, hence exact values of $A$, $B$ and $C$, one can easily compute $R$:
$$
R=\left\{\begin{array}{ll}
\frac{-B\pm\sqrt{B^2-4AC}}{2A} & \textrm{if $A\neq 0$},\\
&\\
-\frac{C}{B} & \textrm{if $A=0$}.
\end{array}\right.
$$
But what if $A$ is an interval which contains zero (and not reduced to $\{0\}$)?
The second above formula cannot be used, while the division by $A$ in the first one yields $R=(-\infty,\infty)$.
To get around that, we make two cases, depending whether the interval $AC$ contains zero or not.

If $AC$ does not contain zero, we simply use the first formula.
More precisely, a short case study shows that the smallest positive root is
$$
R=\frac{-B-\textrm{sign}(C)\sqrt{B^2-4AC}}{2A},
$$
where $\textrm{sign}(C)$ denotes the sign of $C$, which is well defined since the interval $C$ does not contain zero.

If $AC$ contains zero, we shall use the fact that the roots of a polynomial are continuous in its coefficients.
Namely, when $A$ tends towards $0$, one of the roots goes to infinity while the interesting one goes towards $-\tfrac{C}{B}$.
This latter root is
$$
R=\frac{-B+B\sqrt{1-\tfrac{4AC}{B^2}}}{2A}=-\frac{C}{B}\frac{2B^2}{4AC}\left(1-\sqrt{1-\frac{4AC}{B^2}}\right).
$$
With $x=\tfrac{4AC}{B^2}$ and $f(x)=\tfrac{2}{x}(1-\sqrt{1-x})$, this can be written
$$
R=-\tfrac{C}{B}f(x).
$$
If we set $f(0):=1$, then $f$ becomes continuously derivable over $(-\infty,1)$.
The Taylor's theorem then ensures that for any real number $x<1$, there exists a real number $\xi$ between $0$ and $x$ such that
$$
f(x)=1+xf'(\xi).
$$
One checks that $f'$ is positive and increasing over $(-\infty,1)$.
One computes
$$
f'(0.78)\approx 0.9879 <1.
$$
If $x$ is an interval which contains $0$ and whose upper bound is at most $0.78$, then
$$
f(x) \subset 1+x\times f'\left((-\infty,0.78]\right) \subset 1+x\times [0,1]= 1+x.
$$
This yields the wanted interval around $-\tfrac{C}{B}$:
$$
R\subset -\frac{C}{B}\left(1+\frac{4AC}{B^2}\right).
$$
The above formula still yields $R=(-\infty,\infty)$ if $B$ contains $0$ as well as $A$.
Moreover, we assumed $x<0.78$, which can be false if $A$, $B$ and $C$ are large interval: in such a case we have to use the first formula which also yields $R=(-\infty,\infty)$.
Both cases however happens only when the intervals $A$, $B$ and $C$ have a quite large diameter, that is, in the very few first steps of the recursive local inequality checking.

\section{Parameters of the edge potential}
\label{sec:LQ}

The rule of thumb (which could perhaps be made rigorous) used to choose the constants $l_e$ and $q_e$ in Prop.~\ref{prop:local} (Tab.~\ref{tab:LQ}) is that if the local inequality works for the triangles with only one pair of discs which are not tangent, then it seems to work for any triangle.
We thus consider triangles $T$ with circles of size $x$ and $y$ centered on the endpoints of an edge $e$ and vary the length of $e$,  as in Fig.~\ref{fig:plot}.
We first fixed $l_e$ close to the length for which $d_e(T)$ changes its sign, then fixed $q_e$ so that $U(T)$ is slightly less than $E(T)$ when $T$ is stretched.

\paragraph{Acknowledgments.}
We thank Stef Graillat for useful discussions about computing the support disc radius with interval arithmetic when $A$ contains $0$ (Appendix~\ref{sec:support_disc}).
We thank Daria Pchelina for pointing us the necessity of the condition \eqref{eq:mq_for_capping} page~\pageref{eq:mq_for_capping}.
We thank the referee of the first version of this paper.

\bibliographystyle{alpha}
\bibliography{binary}

\begin{thebibliography}{CGSY18}

\bibitem[CGSY18]{CGSY18}
R.~Connelly, S.~Gortler, E.~Solomonides, and M.~Yampolskaya.
\newblock Circle packings, triangulations, and rigidity.
\newblock Oral presentation at the conference for the 60th birthday of Thomas
  C. Hales, 2018.

\bibitem[CW10]{CW10}
H.-C. Chang and L.-C. Wang.
\newblock A simple proof of {T}hue's theorem on circle packing.
\newblock \href{https://arxiv.org/abs/1009.4322}{arxiv:1009.4322}, 2010.

\bibitem[Dev16]{sage}
The~Sage Developers.
\newblock {\em {S}age {M}athematics {S}oftware ({V}ersion 8.2)}, 2016.
\newblock \url{http://www.sagemath.org}.

\bibitem[Fer19]{Fer19}
Th. Fernique.
\newblock A {D}ensest ternary circle packing in the plane.
\newblock \href{https://arxiv.org/abs/1912.02297}{arxiv:1912.02297}, 2019.

\bibitem[FHS20]{FHS20}
Th. Fernique, A.~Hashemi, and O.~Sizova.
\newblock Compact packings of the plane with three sizes of discs.
\newblock {\em Discrete and Computational Geometry}, 2020.

\bibitem[FKS19]{FKS19}
S.~P. Fekete, Ph. Keledenich, and Ch. Scheffer.
\newblock Packing disks into disks with optimal worst-case density.
\newblock In {\em 35th Int. Symp. Comput. Geom., SoCG 2019}, volume 129, pages
  35:1--35:19, 2019.

\bibitem[FP21]{FP21}
Th. Fernique and D.~Pchelina.
\newblock Compact packings are not always the densest.
\newblock \href{https://arxiv.org/abs/2104.12458}{arxiv:2104.12458}, 2021.

\bibitem[FT43]{FT43}
L.~Fejes~T{\'o}th.
\newblock Über die dichteste {K}ugellagerung.
\newblock {\em Mathematische Zeitschrift}, 48:676--684, 1943.

\bibitem[FT64]{FT64}
L.~Fejes~T\'oth.
\newblock {\em Regular figures}.
\newblock International Series in Monographs on Pure and Applied Mathematics.
  Pergamon, Oxford, 1964.

\bibitem[FTM58]{FM58}
L.~Fejes~T\'oth and J.~Moln\'ar.
\newblock Unterdeckung und {Ü}berdeckung der {E}bene durch {K}reise.
\newblock {\em Mathematische Nachrichten}, 18:235--243, 1958.

\bibitem[Hal05]{Hal05}
Th. Hales.
\newblock A proof of the {K}epler conjecture.
\newblock {\em Annals of Mathematics}, 162:1065--1185, 2005.

\bibitem[Hep00]{Hep00}
A~Heppes.
\newblock On the densest packing of discs of radius $1$ and $\sqrt{2}-1$.
\newblock {\em Studia Scientiarum Mathematicarum Hungarica}, 36:433--454, 2000.

\bibitem[Hep03]{Hep03}
A.~Heppes.
\newblock Some densest two-size disc packings in the plane.
\newblock {\em Discrete and Computational Geometry}, 30:241--262, 2003.

\bibitem[Ken04]{Ken04}
T.~Kennedy.
\newblock A densest compact planar packing with two sizes of discs.
\newblock \href{https://arxiv.org/abs/math/0412418}{arxiv:0412418}, 2004.

\bibitem[Ken06]{Ken06}
T.~Kennedy.
\newblock Compact packings of the plane with two sizes of discs.
\newblock {\em Discrete and Computational Geometry}, 35:255--267, 2006.

\bibitem[Lag02]{Lag02}
J.~C. Lagarias.
\newblock Bounds for local density of sphere packings and the {K}epler
  conjecture.
\newblock {\em Discrete and Computational Geometry}, 27:165--193, 2002.

\end{thebibliography}
\end{document}